\documentclass[a4paper,11pt]{amsart}
 \usepackage[arrow,matrix]{xy}
\usepackage{amsmath,amssymb,amscd,bbm,amsthm,mathrsfs, chemarrow}
\usepackage{graphicx}
 \numberwithin{equation}{section}
\newtheorem{thm}{Theorem}[section]
\newtheorem{lem}{Lemma}[section]
\newtheorem{cor}{Corollary}[section]
\newtheorem{pro}{Proposition}[section]

\theoremstyle{definition}

\theoremstyle{remark}

\begin{document}\numberwithin{equation}{section}
\title[The resolution of  Euclidean massless field operators  and the $L^2$ method]
{The resolution of  Euclidean massless field operators of higher spins on $\Bbb R^6$ and the $L^2$ method}
\author{Qianqian Kang}

\address[Qianqian Kang]{Department of Mathematics, Zhejiang International Studies University,  Hangzhou 310012, PR China}\email{qqkang@zisu.edu.cn}

\author{Wei Wang}

\address[Wei Wang]{Department of Mathematics, Zhejiang University, Zhejiang 310027, PR China}\email{wwang@zju.edu.cn}

\author{Yuchen Zhang}

\address[Yuchen Zhang]{Department  of Mathematics, University of Science and Technology of China,  Hefei 230026, PR China}\email{yuchen95@mail.ustc.edu.cn}

 \begin{abstract}
   The resolution of $4$-dimensional  massless field operators of higher spins  was constructed by Eastwood-Penrose-Wells by using the twistor method. Recently physicists are interested in $6$-dimensional physics including the massless field operators of higher spins on Lorentzian space $\Bbb R^{5,1}$. Its Euclidean version $\mathscr{D}_0$ and their function theory are discussed in   \cite{wangkang3}. In this paper, we construct an exact sequence of Hilbert spaces as weighted $L^2$ spaces    resolving $\mathscr{D}_0$:
    $$L^2_\varphi(\Bbb R^6, \mathscr{V}_0)\overset{\mathscr{D}_0}\longrightarrow L^2_\varphi(\Bbb R^6,
\mathscr{V}_1)\overset{\mathscr{D}_1}\longrightarrow L^2_\varphi(\Bbb R^6, \mathscr{V}_2)\overset{\mathscr{D}_2}\longrightarrow
L^2_\varphi(\Bbb R^6, \mathscr{V}_3)\longrightarrow 0,$$
     with suitable operators $\mathscr{D}_l$ and vector spaces $\mathscr{V}_l$. Namely, we can solve $\mathscr{D}_{l}u=f$ in $L^2_\varphi(\Bbb R^6, \mathscr{V}_{l})$  when $\mathscr{D}_{l+1} f=0$ for $f\in L^2_{\varphi}(\Bbb R^6, \mathscr{V}_{l+1})$. This is proved by using the $L^2$ method in the theory of several complex variables, which is a general framework to solve  overdetermined PDEs  under the compatibility condition. To apply this method here, it is necessary to  consider   weighted $L^2$ spaces, an advantage of which  is that any polynomial   is $L^2_{\varphi}$ integrable.  As a corollary, we prove that
$$
        P(\Bbb R^6, \mathscr{V}_0)\overset{\mathscr{D}_0}\longrightarrow P(\Bbb R^6,\mathscr{V}_1)\overset{\mathscr{D}_1}\longrightarrow P(\Bbb R^6, \mathscr{V}_2)\overset{\mathscr{D}_2}\longrightarrow
P(\Bbb R^6, \mathscr{V}_3)\longrightarrow 0$$
       is a resolution, where $P(\Bbb R^6, \mathscr{V}_l)$  is the space of  all $\mathscr{V}_l$-valued polynomials. This provides an analytic way to construct a resolution of a differential operator acting on vector valued polynomials.
 \end{abstract}
\keywords{resolution; Euclidean massless field operator of high spins; the $L^2$ method; overdetermined PDEs; the compatibility condition; differential  complexes.}
   \thanks{
The first author is partially supported by National Nature Science Foundation in China (No. 11801523) and the foundation of Zhejiang International Studies University (No. BD2019B9). The second and third  authors are partially supported by National Nature Science Foundation in China (No. 11971425)}
\maketitle
\section{Introduction}

The resolution of massless field operators  of higher spins  over the complexified Minkowski space $\Bbb C^4$ was constructed by Eastwood-Penrose-Wells  \cite{Eastwood} by using twistor method. The Euclidean version of massless field operator of spin $k/2$ is also called $k$-Cauchy-Fueter operator    (cf. \cite{CMW} \cite{Wang2018} and references therein). Recently physicists are interested in $6$-dimensional physics including the massless field operators of higher spins on Lorentzian space $\Bbb R^{5,1}$ (cf. \cite{Mason,Samann} and references there in). The Euclidean version of these operators are
 \begin{equation*}
 \begin{aligned}
 \mathcal{D}_0 : C^{\infty}(\mathbb R^{6},\odot^{k} \mathbb C^{4}) &\longrightarrow
 C^{\infty}(\mathbb R^{6}, \mathbb C^{4}\otimes \odot^{k-1} \mathbb C^{4}),
 \end{aligned}
 \end{equation*}
 $k=1,2,\ldots$, where $\odot^{p} \mathbb C^{4}$ is  $p$-th symmetric power  of $  \Bbb C^{4}$.
 A $\odot^{k}\Bbb C^{4}$-valued distribution $f$ is called \emph{$k$-monogenic} if it satisfies
$ \mathcal{D}_0 f=0$. In \cite{wangkang3}, we   proved various properties for $k$-monogenic functions, e.g. the existence of  infinite number of $k$-monogenic polynomials.
In order to study $k$-monogenic functions, we need to solve the nonhomogeneous equation
$\mathcal{D}_0 u=f,$
which is overdetermined for $k>1$.  So we need to find the compatibility condition for solvability, and more generally a resolution of $\mathscr{D}_0$. Motivated by   $4$-dimensional and the quaternionic cases (c.f. \cite{Baston,Bures,Colombo,Wang2010} and references therein), a natural candidate of the resolution is
\begin{equation}\label{eq:elliptic complex0}C^\infty (\mathbb{R}^6,V_0)\overset{\mathcal{D}_0 }\longrightarrow C^\infty
(\mathbb{R}^6,V_1)\overset{\mathcal{D}_1 }\longrightarrow C^\infty
(\mathbb{R}^6,V_2)\overset{\mathcal{D}_2 }\longrightarrow C^\infty (\mathbb{R}^6,V_3)\overset{\mathcal{D}_3 }\longrightarrow C^\infty(\mathbb{R}^6,V_4)\longrightarrow 0 \end{equation}
with
 \begin{equation}\label{def:Vl}
V_l:=\odot^{k-l}{\mathbb{C}^4}\otimes\wedge^l\mathbb{C}^4,
\end{equation}
  when  $k\geq 4$. But it is already known \cite[Section 1]{wangkang3} that the image of $\mathcal{D}_0 $ consists of functions only valued in a subspace of $V_1$, the kernel of the {\it contraction}   given by
\begin{equation}\label{def:contraction0}
 \mathscr{C}: \odot^{p}\mathbb{C}^4\otimes\wedge^{q}\mathbb{C}^4\rightarrow\odot^{p-1}\mathbb{C}^4\otimes\wedge^{q-1}\mathbb{C}^4
  \end{equation}
with $p=k-1, q=1$. Denote
 \begin{equation}\label{def:vl}
 \mathscr{V}_l:=\ker \mathscr{C} \mid_{V_l}.
 \end{equation} Then
 $\mathscr{V}_4=\{0\}$ automatically. Denote by $\mathscr D_{l}$ the restriction of $\mathcal{D}_l$ in (\ref{eq:elliptic complex0}) to $C^{\infty}(\Bbb R^6, \mathscr{V}_l)$. We construct the following differential complex:
\begin{equation}\label{elliptic complex}
0\longrightarrow C^\infty (\mathbb{R}^6,\mathscr{V}_0)\overset{\mathscr D_{0}}\longrightarrow C^\infty
(\mathbb{R}^6,\mathscr{V}_1)\overset{\mathscr D_1}\longrightarrow C^\infty
(\mathbb{R}^6,\mathscr{V}_2)\overset{\mathscr D_2}\longrightarrow C^\infty (\mathbb{R}^6,\mathscr{V}_3)\longrightarrow
0,
\end{equation}
and call it \emph{$k$-monogenic complex}, $k=4,5,\ldots$. Note that $\Bbb C^4$ is the spin representation of $\mathfrak{so}(6,\mathbb{C})$ and  $\mathscr{V}_l$ as the contraction of $V_l$ is an irreducible representation of $\mathfrak{so}(6,\mathbb{C})$ (cf. \cite{rep}).

The $L^2$ method is a powerful method to solve $\overline{\partial}$-equation in the theory of several complex variables (cf. e.g. \cite{ChenShaw,Hormander1,Hormander3}). In fact, it is a general framework to solve overdetermined PDEs under the compatibility condition, which is also given by a system of PDEs. The main difficulty to use this method is to prove the corresponding  $L^2$ estimate. It was applied to the $k$-Cauchy-Fueter complex over $\Bbb R^{4n}$ in \cite{Wang2017} and also the Neumann problem associated to the $k$-Cauchy-Fueter complex on $k$-pseudoconvex domains in $\Bbb R^4$ \cite{Wang2018}. The latter case is restricted to dimension $4$ because only over $\Bbb R^4$ the corresponding $L^2$ estimates was proved.

In this paper we consider the weighted $L^2$ estimate of the $k$-monogenic complex as in \cite{Wang2017}. We define an inner product $\langle  \cdot, \cdot\rangle$ on $V_l$ and $\mathscr{V}_l$  induced from $\otimes^{k} \Bbb C^4$. Let $L^2_\varphi(\Bbb R^6, \mathscr{V}_l)$ be the Hilbert space of $\mathscr{V}_l$-valued $L^2_\varphi$-integrable functions on $\Bbb R^6$ with weighted inner product
\begin{equation}
\langle f,h\rangle_\varphi :=\int_{\Bbb R^6}\left\langle f(x),h(x)\right\rangle e^{-\varphi}dx,
\end{equation}with  weight $\varphi=\left|x \right|^2$.
Denote weighted norm
$ \left\|f\right\|_\varphi:=\langle f,f\rangle_\varphi^{\frac{1}{2}} $.

  $Dom (\mathscr{D}_l)$ consists of $f\in L^2_\varphi(\Bbb R^6, \mathscr{V}_l)$ such that   $\mathscr{D}_lf=u$ in the weak sense  for some $u\in L^2_\varphi(\Bbb R^6, \mathscr{V}_{l+1})$, i.e.
 \begin{equation}\label{eq:Dom}
    \langle u, g\rangle_\varphi=\langle f,\Theta_l g\rangle_\varphi
 \end{equation}
 for any $C_0^\infty(\Bbb R^6, \mathscr{V}_{l+1})$, where $\Theta_l $ is the formal adjoint of  $\mathscr{D}_l$. The differential operator
  $\mathscr{D}_l$  defines a linear, closed, densely defined operator from $L^2_\varphi(\Bbb R^6, \mathscr{V}_l)$ to
 $L^2_\varphi(\Bbb R^6, \mathscr{V}_{l+1})$,  which we also denote by $\mathscr{D}_l$.
 Denote by $\mathscr{D}_l^\ast$ the adjoint operator of $\mathscr{D}_l$ between Hilbert spaces $L^2_\varphi(\Bbb R^6, \mathscr{V}_{l+1})$ and $L^2_\varphi(\Bbb R^6, \mathscr{V}_{l})$.
 The sequence
\begin{equation}\label{L2-complex}
\begin{aligned}
 L^2_\varphi(\Bbb R^6, \mathscr{V}_0)\overset{\mathscr{D}_0}\longrightarrow L^2_\varphi(\Bbb R^6,
\mathscr{V}_1)\overset{\mathscr{D}_1}\longrightarrow L^2_\varphi(\Bbb R^6, \mathscr{V}_2)\overset{\mathscr{D}_2}\longrightarrow
L^2_\varphi(\Bbb R^6, \mathscr{V}_3)\longrightarrow 0,
\end{aligned}
\end{equation}
is a complex of Hilbert spaces, i.e., for any $u\in Dom(\mathscr{D}_{l})$ and $\mathscr{D}_{l}u \in Dom(\mathscr{D}_{l+1})$, we have $\mathscr{D}_{l+1}\mathscr{D}_{l}u=0$. To find   solution  to the equation
\begin{equation}\label{equ1:Dg=f}
\mathscr{D}_{l}u=f,
\end{equation}
for $f\in L^2(\Bbb R^6,\mathscr{V}_{l+1})$ satisfying
\begin{equation}\label{equ2:Dg=f}
\mathscr{D}_{l+1}f=0,
\end{equation}
  we consider the \emph{associated Hodge Laplacian operator} $\Box_l: L^2_\varphi(\Bbb R^6, \mathscr{V}_l)\longrightarrow L^2_\varphi(\Bbb R^6, \mathscr{V}_l)$ given by
  \begin{equation}
\Box_l:=\mathscr{D}_{l-1}\mathscr{D}_{l-1}^\ast+\mathscr{D}_{l}^\ast\mathscr{D}_{l},\qquad l=1,2,\quad {\rm and }\quad\Box_3:=\mathscr{D}_{2}\mathscr{D}_{2}^\ast.
\end{equation}

\begin{thm}\label{Thm-solution}
Suppose that $\varphi(x)=\left|x\right|^2$ and $k=6,7,\ldots$.  There exists a constant $C>0$ only depending on $k$ such that\\
$(1)$\,\,\,\,$\Box_l$ has a bounded self-adjoint inverse $N_l$ such that
\begin{equation}
\left\|N_lf\right\|_\varphi\leq C \left\|f\right\|_\varphi \,\,\,\,\,\,\,\,\text{for}\,\,\text{any}\,\,f\in L^2_\varphi(\Bbb R^6,\mathscr{V}_{l}).
\end{equation}
$(2)$ \,\,\,\,$\mathscr{D}_{l}^\ast N_{l+1}f$ is the canonical solution to the nonhomogeneous equation \eqref{equ1:Dg=f}-\eqref{equ2:Dg=f}, i.e.
$$\mathscr{D}_{l}\big(\mathscr{D}_{l}^{\ast} N_{l+1}f\big)=f,$$
if $\mathscr{D}_{l+1}f=0$,  and $\mathscr{D}_{l}^{\ast} N_{l+1}f$ orthogonal to $\ker\mathscr{D}_{l}$. Moreover,
\begin{equation}\label{solution-estimate}
\left\|\mathscr{D}_{l}^\ast N_{l+1}f\right\|_\varphi\leq C\left\|f\right\|_\varphi,\,\,\,\,\,\,\,\,\,\left\|\mathscr{D}_{l+1}N_{l+1}f\right\|_\varphi\leq C\left\|f\right\|_\varphi.
\end{equation}
\end{thm}

\begin{cor}
When  $k=6,7,\ldots$, the sequence \eqref{L2-complex} is exact.
\end{cor}

The key step to prove Theorem \ref{Thm-solution} is to establish the following weighted  $L^2$ estimate.
\begin{thm}\label{Thm:k-monogenic l2e}
Suppose that $\varphi(x)=\left|x\right|^2$ and $k=6,7,\ldots$. There exists a constant $C>0$ only depending on $k$ such that
\begin{equation}\label{k-monogenic l2estimate}
\left\|f\right\|^2_{\varphi}\leq C\bigg(\left\|\mathscr{D}_{l}f\right\|_{\varphi}^2+\left\|\mathscr{D}_{l-1}^*f\right\|^2_{\varphi}\bigg),
\end{equation}
for any $f\in Dom(\mathscr{D}_{l})\cap Dom(\mathscr{D}_{l-1}^\ast)$,  $l=1,2,3$.
\end{thm}
  The restriction $k\geq6$ is a technique requirement since we only prove the estimate \eqref{k-monogenic l2estimate} in this case.
 The massless field operators of higher spins on any dimensional Euclidean space was  introduced by Sou$\breve{c}$ek earlier \cite{soucek,soucek2}.
 We only consider $6$-dimensional case here because we can use spin indices based on $\mathfrak{so}(6, \mathbb{C})\cong \mathfrak{sl}(4, \mathbb{C})$, as two-component notation in dimension $4$ based on $  \mathfrak{so}(4, \mathbb{C})\cong \mathfrak{sl}(2, \mathbb{C})\oplus\mathfrak{sl}(2, \mathbb{C}) $.

An advantage to consider weighted $L^2$ space is that any polynomial on $\Bbb R^6$ is $L^2_{\varphi}$ integrable. This allow us to  deduce a resolution of the operator $\mathscr{D}_0$ on $P(\Bbb R^{6},\mathscr{V}_0)$, the module of $\mathscr{V}_0$-valued polynomials over $\Bbb R^{6}$.
\begin{thm}\label{poly-exact}
When  $k=6,7,\ldots$, the sequence
\begin{equation}\label{poly-sequence}
P(\Bbb R^6, \mathscr{V}_0)\overset{\mathscr{D}_0}\longrightarrow P(\Bbb R^6,
\mathscr{V}_1)\overset{\mathscr{D}_1}\longrightarrow P(\Bbb R^6, \mathscr{V}_2)\overset{\mathscr{D}_2}\longrightarrow
P(\Bbb R^6, \mathscr{V}_3)\longrightarrow 0,
\end{equation}
is exact.
\end{thm}
To prove this theorem, we need  the following proposition.
\begin{pro}\label{pro:elliptic complex}
When   $k=4,5,\ldots$, the complex \eqref{elliptic complex} is an elliptic complex.
\end{pro}

Compared to the quaternionic case \cite{Wang2017}, the main difficulty comes from the algebraic complexity when  passing from vector space $V_l=\odot^{k-l}{\mathbb{C}^4}\otimes\wedge^l\mathbb{C}^4$ to its contraction subspace $\mathscr{V}_l$, although only linear algebra is used to overcome it.
In Section \ref{section2}, we give some basic propositions on   symmetrization and antisymmetrization to handle   functions valued in the subspace $\mathscr{V}_l $.
To write down the formal adjoint operator $\mathscr{D}_l^{\ast}$ explicitly, we introduce the orthogonal projection $\mathscr{P}_l$ from $V_l$ to $\mathscr{V}_l^\perp$.
In Section \ref{section3}, we give the   expression of operator $\mathcal{D}_l$ and formal adjoint operators $\mathcal{D}^\ast_l$. Then we prove the $L^2$ estimate \eqref{k-monogenic l2estimate}. In Section \ref{section4}, we give the canonical solution to the nonhomogeneous equations \eqref{equ1:Dg=f}-\eqref{equ2:Dg=f}  by the general framework to solve nonhomogeneous overdetermined  PDEs. As a corollary, we show the sequence \eqref{L2-complex} is exact. Since the $\mathscr{V}_{l}$-valued polynomials over $\Bbb R^6$ are $L^2_{\varphi}$ integrable,  we show that  \eqref{poly-sequence} is also exact. In Section \ref{section5}, we  establish   the ellipticity of the differential complex \eqref{elliptic complex}     by showing the exactness of its symbol sequence, based on which we show $\Box_l$ is an elliptic differential operator and then prove Theorem \ref{poly-exact}.

\section{Linear algebra for symmetric and exterior forms}\label{section2}

\subsection{Symmetrization and antisymmetrization.}
   An element $\xi\in\otimes^t\mathbb{C}^4$ is a tuple $(\xi_{A_1\ldots A_t})$ with $\xi_{A_1\ldots A_t} \in \Bbb C$, where $A_{1},\ldots,A_{k}=1,2,3,4$.
The {\it symmetric power} $\odot^{k} \mathbb C^{4}$ is a subspace of $\otimes^{k} \Bbb C^{4}$, whose element is
 a tuple $(\xi_{B_{1}\ldots B_{k}})$ such that $\xi_{B_{1}\ldots B_{k}}$ is invariant under permutations of subscripts.  An element of $ \odot^{k-l} \mathbb C^{4}\otimes\wedge^{l}\mathbb C^{4}$ is denoted by $(\eta^{A_1\ldots A_l}_{B_{1}\ldots B_{k-l}})$, where
 $\eta^{A_1\ldots A_l}_{B_{1}\ldots B_{k-l}}$ is symmetric in $B_{1},\ldots,B_{k-l}$ and antisymmetric in  $A_1,\ldots,A_{l}$. The norm $\| \xi\|$
 for $\xi \in  \odot^{k-l} \mathbb C^{4}\otimes\wedge^{l}\mathbb C^{4}$ is the norm of $\xi$ as an element of $\otimes^{k} \mathbb C^{4}$, i.e., $\| \xi\|=\sum_{A_1,\ldots, B_1,\ldots}\|\xi^{A_1\ldots A_l}_{B_{1}\ldots B_{k-l}}\|$.

Symmetrization and antisymmetrization of an element $\xi\in\otimes^t\mathbb{C}^4$ is given by
\begin{equation}\label{def:sym-anti}
\xi_{(A_1\ldots A_t)}:=\frac{1}{t!}\sum_{\sigma\in S_t}\xi_{A_{\sigma_1}\ldots A_{\sigma_t}},\,\,\,\,\,\,\,\xi_{[A_1\ldots A_t]}:=\frac{1}{t!}\sum_{\sigma\in S_t}\epsilon_{1\ldots t}^{\sigma_1\ldots\sigma_t}\xi_{A_{\sigma_1}\ldots A_{\sigma_t}},
\end{equation}
respectively, where $S_t$ denotes the permutation group of $t$ elements and $\epsilon_{1\ldots t}^{\sigma_1\ldots\sigma_t}$ is the sign of the permutation from $(1,\ldots,t)$ to $(\sigma_1,\ldots,\sigma_t)$.
The symmetrization or antisymmetrization  of $\xi\in \otimes^t \Bbb C^4$ is an element of $\odot^t\Bbb C^4$  or $\wedge^t\Bbb C^4$.

\begin{lem}[cf. \cite{Wang2017}]\label{lem:sym}
$(1)$ For any $\xi\in\odot^t\mathbb{C}^4$ and $\zeta\in\otimes^t\mathbb{C}^4$, we have
$$\sum_{B_1,\ldots ,B_t}\xi_{B_1\ldots B_t}\overline{\zeta_{B_1\ldots B_t}}=\sum_{B_1,\ldots ,B_t}\xi_{B_1\ldots B_t}\overline{\zeta_{(B_1\ldots
B_t)}}.$$
$(2)$ For any $\xi\in\wedge^t\mathbb{C}^4$ and $\eta\in\otimes^t\mathbb{C}^4$, we have
$$\sum_{B_1,\ldots ,B_t}\xi_{B_1\ldots B_t}\overline{\zeta_{B_1\ldots B_t}}=\sum_{B_1,\ldots ,B_t}\xi_{B_1\ldots B_t}\overline{\zeta_{[B_1\ldots
B_t]}}.$$
\end{lem}
\begin{proof}
 By  definition \eqref{def:sym-anti}, we have
\begin{equation*}
\begin{aligned}
\sum_{B_1,\ldots ,B_t}\xi_{B_1\ldots B_t}\overline{\zeta_{(B_1\ldots B_t)}}=&\frac{1}{t!}\sum_{B_1,\ldots, B_t}\sum_{\sigma\in S_t}\xi_{B_1\ldots
B_t}\overline{\zeta_{B_{\sigma_1}\ldots B_{\sigma_t}}}\\
=&\frac{1}{t!}\sum_{\sigma\in S_t}\sum_{B_1,\ldots ,B_t}\xi_{B_{\sigma_1}\ldots B_{\sigma_t}}\overline{\zeta_{B_{\sigma_1}\ldots B_{\sigma_t}}}=\sum_{B_1,\ldots ,B_t}\xi_{B_1\ldots B_t}\overline{\zeta_{B_1\ldots B_t}},
\end{aligned}
\end{equation*}
by $\xi\in \odot^t\mathbb{C}^4$ and  relabeling   indices. The proof of $(2)$ is analogous to $(1)$.
\end{proof}
\begin{lem}\label{lem:symzhankai}
$(1)$ For $(\xi_{B_1\ldots B_t})\in\otimes^t\mathbb{C}^4$  symmetric in $B_2\ldots B_t$, we have
$$\xi_{(B_1\ldots B_t)}=\frac{1}{t}\sum_{s=1}^t\xi_{B_sB_1\ldots \widehat{B}_{s}\ldots},$$
where $\widehat{B}_{s}$ means omitting $B_{s}$.\\
$(2)$ For $(\xi_{B_1\ldots B_t})\in\otimes^t\mathbb{C}^4$  antisymmetric in $B_2\ldots B_t$, we have
$$\xi_{[B_1\ldots B_t]}=\frac{1}{t}\Big(\xi_{B_1\ldots B_t}-\sum_{s=2}^t\xi_{B_sB_2\ldots  B_1\ldots }\Big),$$
$(3)$ For $\xi\in\otimes^t\mathbb{C}^4$, we have
$$\xi_{[B_1\ldots B_t]}=\xi_{[[B_1\ldots B_{t_1}]B_{t_1+1}\ldots B_t]}.$$
\end{lem}
\begin{proof}
$(1)$ Its proof  is similar to $(2)$. \par
$(2)$ By the definition of antisymmetrization \eqref{def:sym-anti}, we have
\begin{equation}\label{le:duichenzhaikai_1}
\begin{aligned}
\xi_{[B_1\ldots B_t]}=&\frac{1}{t!}\sum_{\sigma\in S_t}\epsilon_{12\ldots t}^{\sigma_1\ldots\sigma_t}\xi_{B_{\sigma_1}\ldots B_{\sigma_t}} =
 \frac{1}{t!}\sum_{s=1}^t\sum_{\sigma\in S_{t},\sigma_1=s}\epsilon_{12\ldots t}^{s\sigma_2\ldots\sigma_t}
\xi_{B_sB_{\sigma_2}\ldots B_{\sigma_t}}.
\end{aligned}
\end{equation}
Since $f$ is antisymmetric in the last $t-1$ indices, we get
\begin{equation*}
\begin{aligned}
\xi_{[B_1\ldots B_t]}=&\frac{(t-1)!}{t!}\xi_{B_1\ldots B_t}+\frac{1}{t!}\sum_{s=2}^t\sum_{\sigma\in S_{t},\sigma_1=s}\epsilon_{12\ldots
t}^{s\sigma_2\ldots\sigma_t}\epsilon_{s\sigma_2\ldots\sigma_s\ldots\sigma_t}^{s2\ldots 1 \ldots t}
\xi_{B_sB_2\ldots   B_1 \ldots}\\
=&\frac{1}{t}\xi_{B_1\ldots B_t}-\frac{1}{t}\sum_{s=2}^t\xi_{B_sB_2\ldots   B_1 \ldots}.
\end{aligned}
\end{equation*}
 This completes the proof of $(2)$. \par
(3) Denote $\Xi_{B_1\ldots B_t}:=\xi_{[B_1\ldots B_{t_1}]B_{t_1+1}\ldots B_t}$. By definition of antisymmetrization, we have
\begin{equation}\label{lem2.2:anti-expan}
\begin{aligned}
\Xi_{[B_1\ldots B_t]}=&\frac{1}{t!}\sum_{\sigma\in S_{t}}\epsilon_{1\ldots t}^{\sigma_1\ldots\sigma_t}\Xi_{B_{\sigma_1}\ldots B_{\sigma_t}}
=\frac{1}{t!t_1!}\sum_{\tau\in S_{t_1}}\sum_{\sigma\in S_{t}}\epsilon_{1\ldots t}^{\tau_{\sigma_1}\ldots \tau_{\sigma_{t_1}}\ldots \sigma_t}\xi_{B_{\tau_{\sigma_1}}\ldots B_{\tau_{\sigma_{t_1}}\ldots B_{\sigma_t}}}\\
=&\frac{1}{t!t_1!}\sum_{\tau\in S_{t_1}}\sum_{\kappa\in S_t}\epsilon_{1\ldots t}^{\kappa_1\ldots\kappa_t}\xi_{B_{\kappa_1}\ldots B_{\kappa_t}}
=\xi_{[B_1\ldots B_t]},
\end{aligned}
\end{equation}
by relabeling indices and permutations in the third identity.
\end{proof}
\subsection{The orthogonal projection  $\mathscr{P_l}:V_l\rightarrow \mathscr{V}_{l}^\perp$}
The {\it contraction} $\mathscr{C} $ (\ref{def:contraction0})   given by
\begin{equation}\label{def:contraction}
\begin{aligned}
\mathscr{C}(f)_{B_1 \ldots B_{p-1}}^{A_1 \ldots A_{q-1}}=\sum_{C=1}^4 f_{B_1 \ldots B_{p-1}C}^{CA_1 \ldots A_{q-1}},
\end{aligned}
 \end{equation}  satisfies
\begin{equation}\label{pro:contraction}
\mathscr{C}\circ\mathscr{C}f=0.
 \end{equation} This is because
  for any fixed $A_1, \ldots, A_{q-2}, B_1, \ldots, B_{p-2}$, \begin{equation*}
(\mathscr{C}\circ\mathscr{C}f)_{B_1 \ldots B_{p-2}}^{A_1 \ldots A_{q-2}}=\sum_{A}\mathscr{C}(f)_{ B_1 \ldots B_{p-2}A}^{A A_1 \ldots A_{q-2}}=\sum_{A,C}f_{B_1 \ldots B_{p-2}AC}^{C A A_1 \ldots A_{q-2}}=0,
 \end{equation*}
by $f$ symmetric in subscripts $A, C$ and antisymmetric in  superscripts $A,C$.

Let $\mathscr{V}_{l}^{\perp}$ be the orthogonal complement   of $\mathscr{V}_{l}$ in $ V_{l}  $.
Now we construct  a linear transformation $\mathscr{P}_l$  from $V_l$ to $\mathscr{V}_{l}^{\perp}$ ($l=1,2$)
\begin{equation}\label{Pf}
\begin{split}
 \mathscr{P}_1 (f)^A_{B_1B_2\ldots B_{k-1}}:&=\frac{k-1}{k+2}\delta^{A}_{(B_{1}}\mathscr{C}(f)_{B_2\ldots B_{k-1})},\quad\qquad f\in V_{1},\\
 \mathscr{P}_2 (f)^{A_1A_{2}}_{B_{1}\ldots B_{k-2}}:&=\frac{2(k-2)}{k}\delta_{(B_{1}}^{[A_1}\mathscr{C}(f)^{ A_2]}_{B_{2}\ldots
B_{k-2})},\qquad f\in V_{2}.
\end{split}
\end{equation}

\begin{pro}\label{P-orthonogal}
 $\mathscr{P}_l$ is an orthogonal projection from $V_l$ to $\mathscr{V}_{l}^\perp$, $l=1,2$.
 \end{pro}
 \begin{proof}
We only prove the case $l=2$ since it is similar for the case $l=1$. Note that
\begin{equation}\label{CP2f=Cf}
\mathscr{C}(\mathscr{P}_{2}f)=\mathscr{C}(f)
\end{equation}
for any $f\in V_2=\odot^{k-2}{\mathbb{C}^4}\otimes\wedge^2\mathbb{C}^4$, since $\mathscr{C}(\mathscr{P}_{2}f)_{B_1\ldots B_{k-3}}^{A_2}= \sum_{A_1}\mathscr{P}_{2}(f)^{A_1A_2}_{B_1\ldots B_{k-3}A_1} $ equals to
\begin{equation*}
\begin{aligned}
 &\frac{1}{k}\sum_{A_1}\bigg(\delta^{A_1}_{B_{1}}\mathscr{C}(f)^{A_2}_{B_2\ldots B_{k-3}A_1}+\ldots+\delta^{A_1}_{B_{k-3}}\mathscr{C}(f)^{A_2}_{B_1\ldots B_{k-4}A_1}+\delta^{A_1}_{A_1}\mathscr{C}(f)^{A_2}_{B_1\ldots B_{k-3}}\\
&\qquad\quad-\delta^{A_2}_{B_{1}}\mathscr{C}(f)^{A_1}_{B_2\ldots B_{k-3}A_1}-\ldots-\delta^{A_2}_{B_{k-3}}\mathscr{C}(f)^{A_1}_{B_1\ldots B_{k-4}A_1}-\delta^{A_2}_{A_1}\mathscr{C}(f)^{A_1}_{B_1\ldots B_{k-3}}\bigg)\\
=&\mathscr{C}(f)^{A_2}_{B_1\ldots B_{k-3}},
\end{aligned}
\end{equation*}
 by \eqref{def:contraction}, \eqref{pro:contraction}, definition \eqref{Pf} and Lemma \ref{lem:symzhankai} $(1)$. Then $\mathscr{P}_{2}$ is a projection since
 \begin{equation*}
\begin{aligned}
\mathscr{P}_{2}(\mathscr{P}_{2}f)_{B_1\ldots B_{k-2}}^{A_1A_2}
=&\frac{2(k-2)}{k}\delta^{[A_1}_{(B_{1}}(\mathscr{C}\mathscr{P}_{2}f)^{ A_2]}_{B_2\ldots B_{k-2})}
= \frac{2(k-2)}{k}\delta^{[A_1}_{(B_{1}}\mathscr{C}(f)^{A_2]}_{B_2\ldots B_{k-2})}
=\mathscr{P}_{2}(f)^{A_1A_2}_{B_1\ldots B_{k-2}},
\end{aligned}
\end{equation*}
by \eqref{CP2f=Cf}. For any $h \in \mathscr{V}_{2}$, we have
\begin{equation*}
\begin{aligned}
k\langle \mathscr{P}_2f, h\rangle
=&\sum_{A_1, A_2,B_1,\ldots,B_{k-2}}\left\langle\sum_{s=1}^{k-2}\bigg(\delta^{A_1}_{B_{s}}\mathscr{C}(f)^{A_2}_{B_1\ldots \widehat{B}_{s}\ldots B_{k-2}}-\delta^{A_2}_{B_{s}}\mathscr{C}(f)^{A_1}_{B_1\ldots \widehat{B}_{s}\ldots B_{k-2}}\bigg), h^{A_1A_2}_{B_1\ldots B_{k-2}}\right\rangle\\
=&\sum_{s=1}^{k-2}\bigg\{\sum_{A_2,B_{1},\ldots,\widehat{B}_{s},\ldots}\bigg\langle \mathscr{C}(f)^{A_2}_{B_1\ldots \widehat{B}_{s}\ldots B_{k-2}}, \sum_{B_s}h^{B_sA_2}_{B_1\ldots B_{k-2}}\bigg\rangle\\
&\qquad-\sum_{A_1,B_{1},\ldots,\widehat{B}_{s},\ldots}\bigg\langle \mathscr{C}(f)^{A_1}_{B_1\ldots \widehat{B}_{s}\ldots B_{k-2}}, \sum_{B_s}h^{A_1B_s}_{B_1\ldots B_{k-2}}\bigg\rangle\bigg\}=0,
\end{aligned}
\end{equation*}
 by definition \eqref{Pf} and $\mathscr{C}h=0$.   Hence, $\mathscr{P}_{2}f \in \mathscr{V}_{2}^{\perp}$, and so it
 is a projection from $V_2$ to $\mathscr{V}_{2}^\perp$.

For any $f\in \ker \mathscr{P}_{2}$, we know $f\in \mathscr{V}_{2}$ since we have $\mathscr{C}(f)=\mathscr{C}(\mathscr{P}_{2}f)=0$.  On the other hand, for any $f\in \mathscr{V}_{2}$, we know
$\mathscr{P}_{2}f=0$ by definition \eqref{Pf}. Then $f \in \ker\mathscr{P}_{2} $ if and only if $f\in \mathscr{V}_{2}$. Hence $\mathscr{P}_{2}$ is an orthogonal projection from $V_2$ to $\mathscr{V}_{2}^\perp$.
\end{proof}

We also need to know the norm of $\mathscr{P}_l(\xi)$ for $\xi \in V_l$ in the proof of the $L^2$ estimate. We have

\begin{pro}\label{prop:pf=some-cf}
 \begin{equation}\label{pf=some-cf}
 \begin{aligned}
 &\left\|\mathscr{P}_1(\xi)\right\|^2=\frac{k-1}{k+2}\left\|\mathscr{C}(\xi)\right\|^2,\qquad \quad\text{for}\,\,\xi \in V_1,\\
 &\left\|\mathscr{P}_2(\xi)\right\|^2=\frac{2(k-2)}{k}\left\|\mathscr{C}(\xi)\right\|^2,\qquad\text{for}\,\,\xi \in V_2.
  \end{aligned}
  \end{equation}
\end{pro}
\begin{proof} For $\xi \in V_1$, we have
\begin{equation}\label{pf=s1+s2}
\begin{aligned}
(k+2)^2\left\|\mathscr{P}_1(\xi)\right\|^2=&(k+2)^2\sum_{A,B_{1},\ldots ,B_{k-1}}\left\|\mathscr{P}_1(\xi)^{A}_{B_{1}\ldots B_{k-1}}\right\|^2\\
=&\sum_{A,B_{1},\ldots,B_{k-1}}\sum_{j,l=1}^{k-1}\left\langle\delta^{A}_{B_{j}}(\mathscr{C}\xi)_{B_1\ldots \widehat{B}_j\ldots
B_{k-1}}, \delta^{A}_{B_{l}}(\mathscr{C}\xi)_{B_1\ldots \widehat{B}_l\ldots B_{k-1}}\right\rangle\\
=&\sum_{A,B_{1},\ldots,B_{k-1}}\left(\sum_{j=1}^{k-1}\left\langle\delta^{A}_{B_{j}}(\mathscr{C}\xi)_{B_1\ldots \widehat{B}_j\ldots
B_{k-1}}, \delta^{A}_{B_{j}}(\mathscr{C}\xi)_{B_1\ldots \widehat{B}_j\ldots B_{k-1}}\right\rangle\right. \\
&\,\,\,\, +\left.\sum_{j\neq l}\left\langle\delta^{A}_{B_{j}}(\mathscr{C}\xi)_{B_1\ldots \widehat{B}_j\ldots
B_{k-1}}, \delta^{A}_{B_{l}}(\mathscr{C}\xi)_{B_1\ldots \widehat{B}_l\ldots B_{k-1}}\right\rangle\right) :=S_1+S_2,
\end{aligned}
\end{equation} by \eqref{Pf}. But we have
\begin{equation}\label{s1-1}
S_1=4(k-1)\sum_{B_{1},\ldots ,B_{k-2}}\left\|(\mathscr{C}\xi)_{B_1\ldots B_{k-2}}\right\|^2=4(k-1)\left\|\mathscr{C}(\xi)\right\|^2,
\end{equation}
by relabeling indices and
\begin{equation}\label{s2-1}
\begin{aligned}
S_2&=\sum_{j\neq l}\sum_{\ldots, \widehat{B_j},\ldots, \widehat{B_l},\ldots}\sum_{A}\left\langle(\mathscr{C}\xi)_{B_1\ldots \widehat{B}_j\ldots A \ldots
B_{k-1}}, (\mathscr{C}\xi)_{B_1\ldots A \ldots \widehat{B}_l\ldots B_{k-1}}\right\rangle\\
&=(k-1)(k-2)\sum_{B_{1},\ldots ,B_{k-2}}\left\|(\mathscr{C}\xi)_{B_1\ldots B_{k-2}}\right\|^2=(k-1)(k-2)\left\|\mathscr{C}(\xi)\right\|^2,
\end{aligned}
\end{equation}
by relabeling indices.
Then the sum of \eqref{s1-1} and \eqref{s2-1} gives us  the first identity of  \eqref{pf=some-cf}.

For $\xi\in V_2$, we rewrite $k^2\left\|\mathscr{P}_2(\xi)\right\|$ as in \eqref{pf=s1+s2}
\begin{equation}\label{CDf-2=s1+s2}
\begin{aligned}
k^2\left\|\mathscr{P}_2(\xi)\right\|^2&=4\sum_{\substack {A_{1},A_2,\\B_{1},\ldots ,B_{k-2}}}\left\langle
\sum_{j}\delta^{[A_{1}}_{B_{j}}(\mathscr{C}\xi)^{ A_2]}_{B_1\ldots \widehat{B}_j\ldots B_{k-2}},
\sum_{l}\delta^{[A_{1}}_{B_{l}}(\mathscr{C}\xi)^{A_2]}_{B_1\ldots \widehat{B}_l\ldots B_{k-2}}\right\rangle\\
&=4\sum_{j}\sum_{\substack {A_{1},A_2,\\B_{1},\ldots ,B_{k-2}}}\left\langle\delta^{[A_{1}}_{B_{j}}(\mathscr{C}\xi)^{A_2]}_{B_1\ldots \widehat{B}_j\ldots B_{k-2}},\delta^{[A_{1}}_{B_{j}}(\mathscr{C}\xi)^{A_2]}_{B_1\ldots \widehat{B}_j\ldots B_{k-2}}\right\rangle\\
&+4\sum_{j\neq l}\sum_{\substack {A_{1},A_2,\\B_{1},\ldots ,B_{k-2}}}\left\langle\delta^{[A_{1}}_{B_{j}}(\mathscr{C}\xi)^{A_2]}_{B_1\ldots \widehat{B}_j\ldots B_{k-2}},
\delta^{[A_{1}}_{B_{l}}(\mathscr{C}\xi)^{A_2]}_{B_1\ldots \widehat{B}_l\ldots B_{k-2}}\right\rangle:=S_1+S_2,
\end{aligned}
\end{equation}
by \eqref{Pf}.  Since $S_1=S_2=0$ if $A_1=A_2$, we only need to consider the summation over $A_1\neq A_2$.  We have
\begin{equation}\label{S1-l=2}
\begin{aligned}
 S_1&=\sum_{j}\sum_{B_{1},\ldots ,B_{k-2}}\sum_{A_1\neq A_2}\left|
\delta_{B_j}^{A_1}(\mathscr{C}\xi)^{ A_2}_{B_1\ldots \widehat{B}_j\ldots B_{k-2}}-\delta_{B_j}^{A_2}(\mathscr{C}\xi)^{A_1}_{B_1\ldots \widehat{B}_j\ldots B_{k-2}}\right|^2\\
&=\sum_{j}\sum_{B_{1},\ldots, \widehat{B}_j,\ldots, B_{k-2}}\sum_{A_1\neq A_2}\bigg(\left|
(\mathscr{C}\xi)^{A_2}_{B_1\ldots \widehat{B}_j\ldots B_{k-2}}\right|^2+\left|
(\mathscr{C}\xi)^{ A_1}_{B_1\ldots \widehat{B}_j\ldots B_{k-2}}\right|^2\bigg)\\
&=6(k-2)\sum_{A ,B_{1},\ldots ,B_{k-3}}\left|(\mathscr{C}\xi)^{A }_{B_1\ldots  B_{k-3}}\right|^2=6(k-2)\left\|\mathscr{C}(\xi)\right\|^2.
\end{aligned}
\end{equation}
By taking summation over $B_j$ at first and then $B_l$, we see that $S_2$ equals to
\begin{equation}\label{S2-l=2}
\begin{aligned}
&\sum_{j\neq l}\sum_{\substack {A_{1},A_2,\\B_{1},\ldots ,B_{k-2}}}\left\langle
\delta^{A_{1}}_{B_{j}}(\mathscr{C}\xi)^{ A_2}_{ \ldots \widehat{B}_j\ldots B_l\ldots }
-\delta^{A_{2}}_{B_{j}}(\mathscr{C}\xi)^{ A_1}_{ \ldots \widehat{B}_j\ldots B_l\ldots },\delta^{A_{1}}_{B_{l}}(\mathscr{C}\xi)^{A_2}_{ \ldots B_j\ldots\widehat{B}_l\ldots }-\delta^{A_{2}}_{B_{l}}(\mathscr{C}\xi)^{ A_1}_{ \ldots B_j\ldots\widehat{B}_l\ldots }\right\rangle\\
=&\sum_{j\neq l}\sum_{\substack {A_{1},A_2,\\ \ldots ,\widehat{B}_j,\ldots }}\bigg\{\left\langle
(\mathscr{C}\xi)^{A_2}_{B_1\ldots \widehat{B}_j\ldots B_l\ldots },\delta^{A_{1}}_{B_{l}}(\mathscr{C}\xi)^{A_2}_{B_1\ldots A_1\ldots\widehat{B}_l\ldots }\right\rangle
+\left\langle(\mathscr{C}\xi)^{A_1}_{B_1\ldots \widehat{B}_j\ldots B_l\ldots},\delta^{A_{2}}_{B_{l}}(\mathscr{C}\xi)^{A_1}_{B_1\ldots A_2\ldots\widehat{B}_l\ldots}\right\rangle\\
&\qquad\qquad -\left\langle(\mathscr{C}\xi)^{A_2}_{B_1\ldots \widehat{B}_j\ldots B_l\ldots},\delta^{A_{2}}_{B_{l}}(\mathscr{C}\xi)^{A_1}_{B_1\ldots A_1\ldots\widehat{B}_l\ldots}\right\rangle
-\left\langle(\mathscr{C}\xi)^{A_1}_{B_1\ldots \widehat{B}_j\ldots B_l\ldots},\delta^{A_{1}}_{B_{l}}(\mathscr{C}\xi)^{A_2}_{B_1\ldots A_2\ldots\widehat{B}_l\ldots}\right\rangle\bigg\}\\
=&2\sum_{j\neq l}\sum_{\substack {A_1,A_2,\\ \ldots, \widehat{B}_j,\ldots, \widehat{B}_l,\ldots  }}\left\langle
(\mathscr{C}\xi)^{ A_2}_{B_1\ldots \widehat{B}_j\ldots A_1\ldots },(\mathscr{C}\xi)^{A_2}_{B_1\ldots A_1\ldots\widehat{B}_l\ldots }\right\rangle\\
=&2\sum_{j\neq l}\sum_{A_2,B_{1},\ldots ,B_{k-3}}\left\|(\mathscr{C}\xi)^{ A_2}_{B_1\ldots  B_{k-3}}\right\|^2=2(k-3)(k-2)\left\|\mathscr{C}(\xi)\right\|^2.
\end{aligned}
\end{equation}
Here   last two terms in the right hand side of the first identity vanish by  $\mathscr{C}\circ\mathscr{C}f=0$ in \eqref{pro:contraction}. We relabel indices in the forth identity. Apply \eqref{S1-l=2}-\eqref{S2-l=2} to \eqref{CDf-2=s1+s2} to get the second identity of \eqref{pf=some-cf}.
\end{proof}

\section{The $L^2$ estimate }\label{section3}
\subsection{The Euclidean massless field operator}
S$\ddot{a}$mann, Wolf \cite{Samann} and Mason et al. \cite{Mason} used the embedding  $ \Bbb R^{5,1}\hookrightarrow \Bbb M\subseteq\Bbb C^{4\times 4}$:
  \begin{equation}\label{6-dim-emd}
  \begin{aligned}
 (x^{0},x^{1},\ldots,x^{5}) &\longmapsto
 \left(\begin{array}{cccc}
 0 &x^{0}+x^{5} &-x^{3}-ix^{4}&-x^{1}+ix^{2}\\
 -x^{0}-x^{5} &0 &-x^{1}-ix^{2}&x^{3}-ix^{4}\\
 x^{3}+ix^{4} &x^{1}+ix^{2}&0 &-x^{0}+x^{5}\\
 x^{1}-ix^{2} &-x^{3}+ix^{4}&x^{0}-x^{5}&0\\
 \end{array}\right),
 \end{aligned}
 \end{equation}
 to study massless field equation, where $\Bbb {CM}=\wedge^{2}\Bbb C^{4}$ is the space of complex antisymmetric $4\times 4$ matrices of dimension $6$. This embedding is the generalization of the embedding of the  Minkowski space  into $2\times 2$-Hermitian matrix space: $\Bbb R^{3,1} \hookrightarrow  \Bbb C^{2\times 2}$,
\begin{equation*}\label{4-embed}
 \begin{aligned}
 (x^{0},x^{1},x^{2},x^{3}) &\longmapsto
 \left(\begin{array}{cc}
 x^{0}+x^{1} &x^{2}+ix^{3}\\
 x^{2}-ix^{3} &x^{0}-x^{1}
 \end{array}\right).
 \end{aligned}
 \end{equation*}
  The advantage of this embedding is that ones can use two-component notation generalizing Penrose's two-spinor notation \cite{twistor-penrose1,twistor-penrose2} and apply the twistor method to study these operators.
On the other hand, we can  embed    $4$-dimensional Euclidean space, the quaternionic space $\Bbb H$, into a real subspace of $\Bbb C^{4}$ by $\Bbb H\hookrightarrow \Bbb C^{2\times 2}$,
  \begin{equation*}
  \begin{aligned}
 x^{0}+ix^{1}+jx^{2}+kx^{3} &\longmapsto
 \left(\begin{array}{cc}
 x^{0}+ix^{1} &-x^{2}-ix^{3}\\
 x^{2}-ix^{3} &x^{0}-ix^{1}
 \end{array}\right),
 \end{aligned}
 \end{equation*}
 and obtain the elliptic version of the differential operators  corresponding to  massless field equations of higher spins on $\Bbb R^{4}$, which  are called $k$-Cauchy-Fueter operators in \cite{Wang2010}. For the higher-dimensional case, we use the embedding $\Bbb H^{n} \hookrightarrow \Bbb C^{2n\times2}$, and also apply the twistor method to study $k$-Cauchy-Fueter equations,  e.g. to find series expansion of $k$-regular functions on $\Bbb H^{n}$ by Penrose integral formula (cf. \cite{wangkang1}\cite{wangkang2}). Motivated by the quaternionic case, we introduce the embedding of   $6$-dimensional Euclidean space into $\Bbb C^{4\times4}$ in \cite{wangkang3} by $\iota: \Bbb R^{6}\hookrightarrow \wedge^{2}\Bbb C^{4} \subseteq\Bbb C^{4\times4}$ given by
  \begin{equation}\label{embed-1}
  \begin{aligned}
 \Bbb R^6\ni x=(x^{0},x^{1},\ldots,x^{5}) &\longmapsto
 \iota(x)=\left(\begin{array}{cccc}
 0 &ix^{0}+x^{5} &x^{3}+ix^{4}&x^{1}+ix^{2}\\
 -ix^{0}-x^{5} &0 &x^{1}-ix^{2}&-x^{3}+ix^{4}\\
 -x^{3}-ix^{4} &-x^{1}+ix^{2}&0 &-ix^{0}+x^{5}\\
 -x^{1}-ix^{2} &x^{3}-ix^{4}&ix^{0}-x^{5}&0\\
 \end{array}\right).
 \end{aligned}
 \end{equation}
 This is essentially the embedding \eqref{6-dim-emd} with $x^{0}$ replaced by $ix^{0}$, up to conjugate and sign of some terms.
The Euclidean version $\mathcal{D}_0$ of these massless field operators are
 \begin{equation*}\label{k-mono-0}
 \begin{aligned}
 \mathcal{D}_0: C^{\infty}(\mathbb R^{6},\odot^{k} \mathbb C^{4}) &\longrightarrow
 C^{\infty}(\mathbb R^{6}, \mathbb C^{4}\otimes \odot^{k-1} \mathbb C^{4})
 \end{aligned}
 \end{equation*}
with
\begin{equation}\label{expre-D0}
   \mathcal{D}_0(f)^{A}_{B_{2}\ldots B_{k}} :=\sum_{B_1}\nabla^{B_1A}f_{B_{1}B_{2}\ldots B_{k}},
 \end{equation}
  where $\nabla^{AB}$ are complex vector fields and the matrix  $(\nabla^{AB})$ is just the embedding matrix \eqref{embed-1} with the coordinate
  $x^{j}$ replaced by $\partial_{x^{j}}$, i.e.,
   \begin{equation} \label{k-operator2}
 \left(\nabla^{AB}\right):=\left(\begin{array}{cccc}
 0 &i\partial_{x^{0}}+\partial_{x^{5}} &\partial_{x^{3}}+i\partial_{x^{4}} &\partial_{x^{1}}+i\partial_{x^{2}}\\
 -i\partial_{x^{0}}-\partial_{x^{5}} &0 &\partial_{x^{1}}-i\partial_{x^{2}} &-\partial_{x^{3}}+i\partial_{x^{4}}\\
 -\partial_{x^{3}}-i\partial_{x^{4}} &-\partial_{x^{1}}+i\partial_{x^{2}} &0 &-i\partial_{x^{0}}+\partial_{x^{5}}\\
 -\partial_{x^{1}}-i\partial_{x^{2}} &\partial_{x^{3}}-i\partial_{x^{4}} &i\partial_{x^{0}}-\partial_{x^{5}} &0
 \end{array}\right).
  \end{equation}
   Define the differential operator $\mathcal{D}_l:C^\infty (\mathbb{R}^6,V_l)\longrightarrow C^\infty (\mathbb{R}^6,V_{l+1})$   by
 \begin{equation}\label{ope-Dl}
 \begin{aligned}
&(\mathcal{D}_lf)^{A_1 \ldots A_{l+1}}_{B_{2} \ldots B_{k-l}}:=\sum_{B_1}\nabla^{B_1[A_1}f_{B_1B_{2}
\ldots B_{k-l}}^{A_2 \ldots A_{l+1}]}.
\end{aligned}
\end{equation}

\begin{pro}\label{prop:S-Prime-com}
The sequence
\begin{equation}\label{diff-complex}
0\longrightarrow C^\infty (\mathbb{R}^6,V_0)\overset{\mathcal{D}_0}\longrightarrow C^\infty (\mathbb{R}^6,V_1)\overset{\mathcal{D}_1}\longrightarrow
C^\infty (\mathbb{R}^6,V_2)\overset{\mathcal{D}_2}\longrightarrow C^\infty (\mathbb{R}^6,V_3)\overset{\mathcal{D}_3}\longrightarrow C^\infty
(\mathbb{R}^6,V_4)\longrightarrow
0
\end{equation}
is a differential complex, i.e., $\mathcal{D}_{l+1}\mathcal{D}_l=0$.
\end{pro}
\begin{proof}
By  definition, we have
\begin{equation}\label{composition}
\begin{aligned}
(\mathcal{D}_{l+1}\mathcal{D}_l f)^{A_1 \ldots A_{l+2}}_{B_3 \ldots B_{k-l}} =&
\sum_{B_{2}}\nabla^{B_{2}[A_1}(\mathcal{D}_l f)_{B_2 \ldots B_{k-l}}^{A_2 \ldots
A_{l+2}]}
=\sum_{B_{1},B_{2}}\nabla^{B_{2}[A_1}\nabla^{|B_{1}|[A_2}f_{B_1 \ldots B_{k-l}}^{A_3 \ldots
A_{l+2}]]}\\
=&\sum_{B_{1},B_{2}}\nabla^{B_{2}[[A_1}\nabla^{|B_{1}|A_2]}f_{B_1 \ldots B_{k-l}}^{A_3 \ldots
A_{l+2}]},
\end{aligned}
\end{equation}
by using Lemma \ref{lem:symzhankai} (3). Here $[\ldots |\ldots|\ldots]$ means we do not antisymmetrize indices inside $ | \,\,\,\,  | $.
Note that $\nabla^{BC}$ commutates with $\nabla^{DA}$ since they are differential operators with constant coefficients. So
\begin{equation}\label{antisym-AB}
\begin{aligned}
2\sum_{B,D}\nabla^{B[A}\nabla^{|D|C]}f_{ \ldots BD}^{ \ldots }
=&\sum_{B,D}\big(\nabla^{BA}\nabla^{DC}-\nabla^{BC}\nabla^{DA}\big)f_{ \ldots BD}^{ \ldots }\\
=&\sum_{B,D}\nabla^{BA}\nabla^{DC}f_{ \ldots BD}^{ \ldots }-\sum_{B,D}\nabla^{BA}\nabla^{DC}f_{ \ldots DB}^{ \ldots }=0,
\end{aligned}
\end{equation}
by relabeling $B$ and $D$ in the second identity and $f$ symmetric in $B$ and $D$.  \eqref{composition} vanishes by \eqref{antisym-AB}.
\end{proof}

Recall that $\mathscr C (\mathcal{D}_0 f)=0$ for any $f\in C^1(U,\mathscr{V}_0)$ (cf. \cite[Introduction]{wangkang3}). This fact is true in general.

\begin{pro}\label{prop:S-complex}
 $(1)$ For $f\in C^\infty (\mathbb{R}^6,\mathscr{V}_l)$, we have  $\mathcal{D}_l f\in C^\infty(\mathbb{R}^6,\mathscr{V}_{l+1})$, $l=0,1,2$;\\
 $(2)$ $\mathscr{V}_4=\{0\}$.
\end{pro}
\begin{proof}
$(1)$ This is because
\begin{equation}\label{3.2-2}
\begin{aligned}
\mathscr C (\mathcal{D}_l f)_{B_1 \ldots B_{k-l-2}}^{A_1 \ldots A_{l}}=&\sum_{C}\mathcal{D}_l (f)_{B_1 \ldots B_{k-l-2}C}^{CA_1 \ldots A_{l}}=\sum_{C,B}\nabla^{B[C}f_{B_1\ldots B_{k-l-2}CB}^{A_1 \ldots A_{l}]} \\
=&\frac{1}{l+1}\sum_{C,B}\Big(\nabla^{BC}f_{B_1 \ldots B_{k-l-2}CB}^{A_1 \ldots A_{l}}-\sum_{s=1}^{l}\nabla^{B A_s}f_{B_1 \ldots B_{k-l-2}CB}^{A_1 \ldots C\ldots  A_{l}}\Big)=0,
\end{aligned}
\end{equation}
by using \eqref{def:contraction}, Lemma \ref{lem:symzhankai} (2), $\mathscr C f=0$ and $\nabla^{BC}$ antisymmetric in $B$ and $C$ while $f$ symmetric in $B$ and $C$.

$(2)$ For $f=(f^{A_1A_2A_3A_4}_{B_1\ldots B_{k-4}})\in \mathscr{V}_4$, it is obvious that  $f^{A_1A_2A_3A_4}_{B_1\ldots B_{k-4}}\neq 0$ only if $\{A_1, A_2, A_3, A_4\}=\{1,2,3,4\}$ and so $B_j$ must equal to one of $A_1,\ldots,A_4$. Without loss of generality, we assume $A_1=B_{k-4}$. It follows from $\mathscr C (f)=0$ that
$$f^{A_1A_2A_3A_4}_{B_1\ldots B_{k-5}A_1}=\sum_{C}f^{CA_2A_3A_4}_{B_1\ldots B_{k-5}C}=0,$$
for any fixed $A_2, \ldots, A_{4},B_1, \ldots, B_{k-5}$.
 So $f=0$.
\end{proof}
Since $\mathscr{D}_l$ is the restriction of $\mathcal{D}_l$ on $C^\infty (\mathbb{R}^6,\mathscr{V}_l)$, we have the following corollary.
\begin{cor}
\eqref{elliptic complex} is a differential complex.
\end{cor}

\subsection{The formal adjoint operators}

Let $\epsilon_{ABCD}=\epsilon^{ABCD}$ be the sign of the permutation from $(1,2,3,4)$ to $(A, B,C, D)$. Then $\epsilon_{ABCD}$ vanishes if $\{A,B,C,D\}\neq \{1,2,3,4\}$. We  use $\epsilon^{ABCD}$ and $\epsilon_{ABCD}$ to raise and low indices respectively. For example,
 \begin{equation*}\label{lashang}
 \nabla_{AB}:=\frac{1}{2}\sum_{C,D}\epsilon_{ABCD}\nabla^{CD}.
 \end{equation*}
 Then we have
\begin{equation*}\label{lashang-2}
\nabla^{AB}=\frac{1}{2}\sum_{C,D=1}^{4}\epsilon^{ABCD}\nabla_{CD},
 \end{equation*}
 since $\sum_{C,D}\epsilon_{ABCD}\epsilon^{CDEF}=2(\delta_{A}^{E}\delta_{B}^{F}-\delta_{A}^{F}\delta_{B}^{E})$ by definition (cf. \cite[P.6]{Mason}).
 We know that \cite [Proposition 2.1]{wangkang3} the operators $\nabla^{AB}$ and $\nabla_{AB}$ defined above satisfy
\begin{equation}\label{d1-delta1}
\overline{\nabla_{AB}}=\nabla^{AB}\,\,\,\,\,\,\text{and}\,\,\,\,\,\,\,\sum_{A}\nabla_{AB_{1}}\nabla^{AB_{2}}=\delta_{B_{1}}^{B_{2}}\Delta.
\end{equation}

Let $\Theta_{AB}$ be a scalar  differential operator defined by
\begin{equation}\label{Theta-AB}
\Theta_{AB}f:=-e^\varphi\nabla_{AB}(e^{-\varphi}f)=-\nabla_{AB}f+\big(\nabla_{AB}\varphi\big)f.
\end{equation}

\begin{lem}\label{fenbujifen}
The formal adjoint of the scalar differential operator $\nabla^{AB}$ is $\Theta_{AB}$.
\end{lem}

\begin{proof}
For any $u,v\in C_0^1(\mathbb{R}^6,\mathbb{C})$, we have
\begin{eqnarray*}
\langle\nabla^{AB}u,v\rangle_\varphi&=&\int\nolimits_{\mathbb{R}^6}(\nabla^{AB}u)\bar{v}e^{-\varphi}dV
=\int\nolimits_{\mathbb{R}^6}\bigg\{\nabla^{AB}(u\bar{v}e^{-\varphi})-u\cdot\overline{\nabla_{AB}(v e^{-\varphi})}\bigg\}dV\\
&=& \int\nolimits_{\mathbb{R}^6}u\cdot\overline{\Theta_{AB}v}e^{-\varphi}dV=\langle u,\Theta_{AB}v\rangle_\varphi,
\end{eqnarray*}
by using  Stocks' formula.
\end{proof}

\begin{pro}\label{prop:ad-Dl}
 The formal adjoint $\mathcal{D}_l^{\ast}$ of $\mathcal{D}_l$ is given by
\begin{equation}\label{ad-Dl}
(\mathcal{D}_l^{\ast} f)_{B_{1} \ldots B_{k-l}}^{A_{1}\ldots A_l}=-\sum_{E}\Theta _{E(B_{1}} f_{B_{2} \ldots B_{k-l})}^{E A_{1}\ldots A_{l}},
\end{equation}
for $f\in C_0^{\infty}(\Bbb R^6, V_{l+1})$, $l=0,1,2.$
\end{pro}
\begin{proof}
By definition of the formal adjoint operator, for any $h\in C_0^\infty(\mathbb{R}^6,V_{l+1})$, we have
\begin{equation*}
\begin{aligned}\langle f,\mathcal{D}_{l}^{\ast}h\rangle_\varphi=
&\langle \mathcal{D}_{l} f,h\rangle_{\varphi}
=\sum_{\substack {A_{1},\ldots ,A_{l+1}, \\B_{2},\ldots ,B_{k-l}}}
\left\langle\sum_{B_{1}}\nabla^{B_{1}[A_{1}} f_{B_{1} \ldots B_{k-l}}^{A_{2}\ldots A_{l+1}]},h_{B_2 \ldots B_{k-l}}^{A_1 \ldots A_{l+1}}\right\rangle_\varphi\\
=& \sum_{A_{1},\ldots,B_{1},\ldots }
\left\langle \nabla^{ B_{1}A_{1}} f_{B_{1} \ldots B_{k-l}}^{A_{2}\ldots A_{l+1}},h_{B_2 \ldots B_{k-l}}^{A_1 \ldots A_{l+1}}\right\rangle_\varphi
=\sum_{A_{1},\ldots,B_{1},\ldots }\left\langle f_{B_{1} \ldots B_{k-l}}^{A_{2}\ldots A_{l+1}},
\Theta_{B_{1}A_{1}}h_{B_2 \ldots B_{k-l}}^{A_1 \ldots A_{l+1}}\right\rangle_\varphi\\
=&\sum_{A_{1},\ldots,B_{1},\ldots }\left\langle f_{B_{1} \ldots B_{k-l}}^{A_{2}\ldots A_{l+1}},
     -\sum_{A_{1}}\Theta_{A_{1}(B_{1}}h_{B_2 \ldots B_{k-l})}^{A_1 \ldots A_{l+1}}\right\rangle_\varphi
\end{aligned}
\end{equation*}
by using Lemma \ref{lem:sym} twice and  Lemma \ref{fenbujifen}. The result follows.
\end{proof}

\begin{lem}\label{lem:jiaohuanzi}
For $\varphi=\left|x\right|^2$, we have
$$[\nabla^{AB},\Theta_{CD}]=8\delta^A_{[C}\delta^B_{D]}.$$
\end{lem}
\begin{proof}
 Since $\nabla^{AB},\nabla_{CD}$ are scalar differential operators with constant coefficients,  we have
\begin{equation}\label{delt-theta}
\begin{aligned}
&[\nabla^{AB},\Theta_{CD}]
=[\nabla^{AB},-\nabla_{CD}+\nabla_{CD}\varphi]
= [\nabla^{AB},\nabla_{CD}\varphi]=\nabla^{AB}\nabla_{CD}\varphi,
\end{aligned}
\end{equation}
by \eqref{Theta-AB}. Denote
\begin{equation*}
\begin{aligned}
 z^{AB}:=\left(
 \begin{matrix}
   0 & ix_0+x_5 & x_3+ix_4 & x_1+ix_2 \\
   -ix_0-x_5 & 0 & x_1-ix_2 & -x_3+ix_4 \\
   -x_3-ix_4 & -x_1+ix_2 & 0 & -ix_0+x_5 \\
   -x_1-ix_2 & x_3-ix_4 &ix_0-x_5 & 0
  \end{matrix}
  \right).
\end{aligned}
\end{equation*}Note that for any fixed $D\in\{1,2,3,4\}$, we have
$\varphi=\sum_{E}z^{ED}\overline{z^{ED}} $, and
\begin{equation}\label{lem4.2-1}
\nabla_{CD}z^{ED}=2\delta_C^E,\quad \nabla_{CD}\overline{z^{ED}}=0,\quad
\nabla_{AB}z^{CD}=2(\delta^C_A\delta^D_B-\delta^D_A\delta^C_B),
\end{equation}by \cite [Lemma 2.1]{wangkang3},
from which  we see that
\begin{equation}\label{eq:delta-phi}
\begin{aligned}
\nabla_{CD}\varphi=\sum_{E}(\nabla_{CD}z^{ED}\cdot\overline{z^{ED}}+z^{ED}\cdot\nabla_{CD}\overline{z^{ED}})
=\sum_{E}2\delta_C^E\cdot\overline{z^{ED}}=2\overline{z^{CD}}.
\end{aligned}
\end{equation}
 Apply (\ref{eq:delta-phi}) to (\ref{delt-theta}) to gfet
\begin{eqnarray*}
[\nabla^{AB},\Theta_{CD}]= &2\nabla^{AB}\overline{z^{CD}}
=2\overline{\nabla_{AB}z^{CD}}
=4(\delta^C_A\delta^D_B-\delta^D_A\delta^C_B).
\end{eqnarray*}
The lemma is proved.
\end{proof}

By  Proposition \ref{prop:ad-Dl}, we know $\Theta_l f$ belongs to $C_0^{\infty}(\Bbb R^6, V_l)$ for $f\in C_0^{\infty}(\Bbb R^6, \mathscr{V}_{l+1})$,
and so $\Theta_lf-\mathscr{P}(\Theta_lf) \in C_0^{\infty}(\Bbb R^6, \mathscr{V}_l)$.
 Then the formal adjoint $\Theta_l:C_0^{\infty}(\Bbb{R}^6,\mathscr{V}_{l+1})\rightarrow C_0^{\infty}(\Bbb{R}^6,\mathscr{V}_l)$ of $\mathscr{D}_l$    satisfies
\begin{equation}\label{expre:ad-D1}
\begin{aligned}
\Theta_l f=\mathcal{D}_l^{\ast}f-\mathscr{P}(\mathcal{D}_l^{\ast} f).
\end{aligned}
\end{equation}
This is because for $f\in C_0^{\infty}(\Bbb R^6, \mathscr{V}_{l+1})$,
\begin{equation*}
\begin{aligned}
\left\langle h,\Theta_l f\right\rangle_{ \varphi }
=\left\langle\mathscr{D}_l h,f\right\rangle_{\varphi }
=\left\langle\mathcal{D}_l h,f\right\rangle_{\varphi }
=\left\langle h,\mathcal{D}_l^{\ast} f\right\rangle_{\varphi }
=\left\langle h,\mathcal{D}_l^{\ast}f-\mathscr{P}(\mathcal{D}_l^{\ast} f)\right\rangle_{\varphi },
\end{aligned}
\end{equation*}
for any $h\in C_0^{\infty}(\Bbb R^6, \mathscr{V}_l)$,   by using Proposition \ref{prop:S-complex} (1) and Proposition \ref{P-orthonogal}.
\subsection{Proof of the $L^2$ estimate }It is a well known fact that differential operator  $\mathscr {D}_l: L^2_\varphi(\Bbb R^6,\mathscr V_l)\rightarrow L^2_\varphi(\Bbb R^6,\mathscr V_{l+1})$ defines a linear, closed, densely defined operator.
\begin{lem}\label{lemm:cutoff}
 Suppose that $\eta_n\in C_0^\infty (\Bbb R^6, \Bbb R)$ with $\eta_n\equiv 1$ on $B(0,n)$, supp $\eta_n\subset B(0,n+2)$ and $|{\rm grad}\, \eta_n|\leq 1$. For $f\in  Dom(\mathscr{D}_{l}) \cap Dom(\mathscr{D}_{l-1}^\ast)$,  we have $ \eta_n f\in  Dom(\mathscr{D}_{l}) \cap Dom(\mathscr{D}_{l-1}^\ast)$ and
  \begin{equation*}
    \|  f - \eta_n f)\|_\varphi+ \|\mathscr{D}_{l}(  f)-\mathscr{D}_{l}( \eta_n f)\|_\varphi+ \|\mathscr{D}_{l-1}^\ast( f)-\eta_n\mathscr{D}_{l-1}^\ast(\eta_n f)\|_\varphi\rightarrow0, \qquad {\rm as}\quad n\rightarrow +\infty.
  \end{equation*}
\end{lem}
\begin{proof} Let $u=\mathscr{D}_{l}f$ in the weak sense.   In definition (\ref{eq:Dom}) the operator $\Theta_l$ can be replaced by the formal adjoint operator  $ \mathcal{D}_{l}^{\ast}$ by (\ref{expre:ad-D1}). Then
   \begin{equation*}\begin{split}  \langle \eta_nf,\Theta_l g\rangle_\varphi&=\langle f,\eta_n\mathcal{D}_{l}^{\ast} g\rangle_\varphi=\langle f,\mathcal{D}_{l}^{\ast} (\eta_n g)\rangle_\varphi- \sum\left\langle f_{B_{1} \ldots B_{k-l} }^{  A_{1}\ldots A_{l}},-\sum_{E}\nabla _{E B_{1}}\eta_n g_{B_{2} \ldots B_{k-l} }^{E A_{1}\ldots A_{l}}\right\rangle_\varphi\\&
   =\langle \eta_n u , g \rangle_\varphi+ \sum\left\langle \nabla^{E B_{1}}\eta_n f_{B_{1} \ldots B_{k-l} }^{  A_{1}\ldots A_{l}},  g_{B_{2} \ldots B_{k-l} }^{E A_{1}\ldots A_{l}}\right\rangle_\varphi
          \end{split} \end{equation*}for any $g\in C_0^\infty(\Bbb R^6, \mathscr{V}_{l+1})$. Consequently, we have
    \begin{equation*}\begin{split}
       \|\mathscr{D}_{l}(\eta_n f)-\eta_n\mathscr{D}_{l}( f)\|_\varphi&=\sum_{E,A_1,\ldots, B_2,\ldots}\int_{\Bbb R^6}\left |\sum_{B_1}\nabla^{B_{1}[E}\eta_n f_{B_{1} \ldots B_{k-l} }^{  A_{1}\ldots A_{l}]}\right|^2 e^{-\varphi}dV\\
       &\leq C\sum_{E,A_1,\ldots, B_2,\ldots}\int_{B(0,n+2)\setminus B(0,n )}\left | f_{B_{1} \ldots B_{k-l} }^{  A_{1}\ldots A_{l}}\right|^2 e^{-\varphi}dV\rightarrow 0,
    \end{split} \end{equation*}for some absolute constant $C>0$. Similarly, we have
    \begin{equation*}\begin{split}\left( \mathscr{D}_{l-1}^\ast(\eta_n f)-\eta_n\mathscr{D}_{l-1}^\ast( f)\right)_{B_{0} \ldots B_{k-l } }^{  A_{2}\ldots A_{l}}=\sum_E
    \nabla_{E (B_{0}}\eta_n f_{B_{1} \ldots B_{k-l }) }^{ E A_{2}\ldots A_{l}},
          \end{split} \end{equation*}
          and so $\|\mathscr{D}_{l-1}^\ast(\eta_n f)-\eta_n\mathscr{D}_{l-1}^\ast( f)\|_\varphi\rightarrow0$. The result  follows.
\end{proof}

{\it Proof of Theorem \ref{Thm:k-monogenic l2e}}. We only need to prove the estimate \eqref{k-monogenic l2estimate} for any $f\in C_0^\infty(\mathbb{R}^6,\mathscr{V}_{l})$. This is because we can assume $f\in L^2_{\varphi}(\Bbb R^6,\mathscr{V}_{l})$ is compactly supported by Lemma \ref{lemm:cutoff}, and can check by definition that  $\delta$ regularization $f_\delta=f*\psi_\delta$,  for nonnegative $\psi\in C_0^\infty (\Bbb R^6, \Bbb R)$ with supp $\psi\subset B(0,1)$ and $\int \psi= 1$, satisfies
 \begin{equation*}
    \|  f -  f_\delta  \|_\varphi+ \|\mathscr{D}_{l}(  f)-\mathscr{D}_{l}( f_\delta)\|_\varphi+ \|\mathscr{D}_{l-1}^\ast( f)- \mathscr{D}_{l-1}^\ast(f_\delta)\|_\varphi\rightarrow0, \qquad {\rm as}\quad \delta\rightarrow 0.
  \end{equation*}

$(1)$ For $l=1$, noting that $\mathscr{D}_0^\ast=\mathcal{D}_0^{\ast}$ for any $f\in C^{\infty}_{0}(\Bbb R^6, \mathscr{V}_1)$, we have
\begin{equation}\label{gonge_fanshu}
\begin{aligned}
k\left\| \mathscr{D}_{0}^\ast f \right\|_{\varphi}^2 =& k\sum_{B_{1}, \ldots ,B_{k}} \left\langle
\sum_{C}\Theta_{C (B_{1}}f_{B_2 \ldots B_{k})}^{C},\sum_{D}\Theta_{D (B_{1}}f_{B_2 \ldots B_{k})}^{D} \right\rangle_\varphi\\
=& k\sum_{C,D,B_{1}, \ldots ,B_{k}}\left\langle \Theta_{C (B_{1}}f_{B_2 \ldots B_{k})}^{C},\Theta_{D B_{1}}f_{B_2 \ldots B_{k}}^{D} \right\rangle_\varphi\\
=& \sum_{C,D,B_{1}, \ldots ,B_{k}} \left\langle \Theta_{CB_{1}} f_{B_{2} \ldots B_{k}}^{C},\Theta_{DB_{1}} f_{B_{2} \ldots B_{k}}^{D} \right\rangle_\varphi\\
&+\sum_{s=2}^{k}\sum_{C,D,B_{1},\ldots,B_{k}}\left\langle\Theta_{CB_{s}}f_{B_{1}\ldots\widehat{B_{s}}\ldots B_{k}}^{C},\Theta_{DB_{1}}f_{B_{2}\ldots B_{k}}^{D} \right\rangle_\varphi
:=\Sigma_1+\Sigma_2,
\end{aligned}
\end{equation}
 by using \eqref{ad-Dl} and Lemma \ref{lem:sym} twice. We find that
\begin{equation}\label{Sigma_1}
\Sigma_1=\sum_{B_{1}, \ldots ,B_{k}}\left\|\sum_{C}\Theta_{CB_{k}} f_{B_{1} \ldots B_{k-1}}^{C}\right\|_\varphi^2\geq0,
\end{equation}
To handle $\Sigma_2$,   take adjoint and use commutators to change the order of operators to get
\begin{eqnarray}\label{Sigma_2}
\Sigma_2
&=&\sum_{s=2}^{k}\sum_{C,D,B_{1},\ldots,B_{k}}\left\langle\nabla^{DB_{1}}\Theta_{CB_{s}}f_{B_{1}\ldots\widehat{B_{s}}\ldots B_{k}}^C,f_{B_{2}\ldots B_{k}}^D \right\rangle_\varphi\nonumber\\
&=&\sum_{s=2}^{k} \sum_{C,D,B_{1}, \ldots ,B_{k}} \left\langle
\Theta_{CB_{s}} \nabla^{DB_{1}} f_{B_{1} \ldots \widehat{B_{s}} \ldots B_{k}}^C,f_{B_{2} \ldots B_{k}}^D \right\rangle_\varphi\nonumber\\
&\ &+\sum_{s=2}^{k} \sum_{C,D,B_{1}, \ldots ,B_{k}} \left\langle
\bigg[\nabla^{DB_{1}},\Theta_{CB_{s}}\bigg]f_{B_{1} \ldots \widehat{B_{s}} \ldots B_{k}}^C,f_{B_{2} \ldots B_{k}}^D \right\rangle_\varphi\nonumber:=\Sigma_3+\Sigma_4,
\end{eqnarray}
and
 \begin{equation}\label{Sigma_4}
 \begin{aligned}
&\Sigma_4=\sum_{s=2}^{k} \sum_{C,D,B_{1}, \ldots ,B_{k}}\left\langle
8\delta^{D}_{[C}\delta^{B_{1}}_{B_{s}]}f_{B_{1} \ldots \widehat{B_{s}} \ldots B_{k}}^{C},
     f_{B_{2} \ldots B_{k}}^{D} \right\rangle_\varphi\\
&=4\sum_{s=2}^{k} \sum_{C,B_{2}, \ldots ,B_{k}}\left\langle f_{B_{2} \ldots B_{k}}^{C},f_{B_{2} \ldots B_{k}}^{C} \right\rangle_\varphi
-4\sum_{s=2}^{k} \sum_{B_{1}, \ldots ,B_{k}}\left\langle f_{B_{1} \ldots \widehat{B_{s}} \ldots B_{k}}^{B_{1}},f_{B_{2} \ldots B_{k}}^{B_s}
\right\rangle_\varphi=4(k-1)\left\| f \right\|_{\varphi}^{2},
\end{aligned}
\end{equation}
by using Lemma \ref{lem:jiaohuanzi} and $\mathscr{C}f=0$. This term is the main term that we need to control. To  control $\Sigma_3$,  let us isolate the term concerning $\mathscr{D}_1f$, note that
\begin{eqnarray*}
\Sigma_3&=&\sum_{s=2}^{k} \sum_{C,D,B_{1}, \ldots ,B_{k}} \left\langle
 \nabla^{DB_{1}} f_{B_{1} \ldots \widehat{B_{s}} \ldots B_{k}}^C,\nabla^{CB_{s}}f_{B_{2} \ldots B_{k}}^D \right\rangle_\varphi\\
&=& \sum_{s=2}^{k} \sum_{C,D,B_{1}, \ldots ,B_{k}} \left\langle
 \nabla^{DB_{1}} f_{B_{1} \ldots \widehat{B_{s}} \ldots B_{k}}^C,-2\nabla^{B_{s}[C}f_{B_{2} \ldots B_{k}}^{D]} \right\rangle_\varphi\\
&\ &+\sum_{s=2}^{k} \sum_{C,D,B_{1}, \ldots ,B_{k}} \left\langle
 \nabla^{DB_{1}} f_{B_{1} \ldots \widehat{B_{s}} \ldots B_{k}}^C,\nabla^{DB_{s}}f_{B_{2} \ldots B_{k}}^C \right\rangle_\varphi:=\Sigma_{31}+\Sigma_{32},
\end{eqnarray*}
 by using Lemma \ref{fenbujifen} and $\nabla^{AB}$ antisymmetric in $A, B$.   We see that
\begin{eqnarray*}
\Sigma_{32}=\sum_{s=2}^{k} \sum_{C,D,D_{1},\ldots,D_{k-2}}\left\|\sum_{E}\nabla^{DE} f_{ED_{1} \ldots \ldots
D_{k-2}}^C\right\|_\varphi^2\geq0,
\end{eqnarray*}
by relabeling indices and
\begin{eqnarray*}
\Sigma_{31}
&=&-2\sum_{s=2}^{k} \sum_{C,D,B_{1}, \ldots ,B_{k}} \left\langle \nabla^{B_{1}[C} f_{B_{1} \ldots \widehat{B_{s}} \ldots B_{k}}^{D]},
         \nabla^{B_{s}[C}f_{B_{2} \ldots B_{k}}^{D]}\right\rangle_\varphi\\
&=&-2\sum_{s=2}^{k} \sum_{C,D,B_{2}, \ldots \widehat{B_{s}},\ldots,B_{k}} \left\langle\mathscr{D}_1(f)^{CD}_{B_{2} \ldots \widehat{B_{s}} \ldots B_{k}},
         \mathscr{D}_1(f)^{CD}_{B_{2} \ldots \widehat{B_{s}} \ldots B_{k}}\right\rangle_\varphi=-2(k-1)\left\|\mathscr{D}_1f\right\|^2_\varphi,
\end{eqnarray*}
by using Lemma \ref{lem:sym}. So we get
\begin{equation}\label{sigma3}
\Sigma_3\geq-2(k-1)\left\|\mathscr{D}_1f\right\|^2_\varphi.
\end{equation}
By (\ref{gonge_fanshu})-(\ref{sigma3}), we get the $L^2$ estimate for $l=1$:
$
\left\| f \right\|_{\varphi}^2\leq \frac{k}{4(k-1)}\left\|\mathscr{D}_{0}^\ast f \right\|_{\varphi}^2  +\frac{1}{2}\left\| \mathscr{D}_{1}f \right\|_{\varphi}^2.
$

$(2)$ For $l=2$, the proof  is similar, but is  more complicated, because we have to use projection $\mathscr{P}_l$. Note that
\begin{equation}\label{depos-1}
\begin{aligned}
\left\|\mathcal{D}_{1}^{\ast} f\right\|^2_\varphi&=\left\langle\mathcal{D}_{1}^{\ast} f,\mathcal{D}_{1}^{\ast} f\right\rangle_\varphi
=\left\|\mathscr{D}_{1}^{\ast} f\right\|^2_\varphi+\left\|\mathscr{P}_1(\mathcal{D}_{1}^{\ast} f)\right\|^2_\varphi,
\end{aligned}
\end{equation}
by \eqref{expre:ad-D1}. Apply Proposition \ref{pf=some-cf} to
 $\left\|\mathscr{P}_1\mathcal{D}_{1}^{\ast} (f)\right\|_{\varphi}^2$ to get
\begin{equation}\label{CDf-1}
\begin{split}
\left\| \mathscr{P}_1(\mathcal{D}_{1}^{\ast} f)\right\|_{\varphi}^2
=&\frac{k-1}{k+2}\sum_{B_{1},\ldots ,B_{k-2}}\left\|\sum_{E}\mathcal{D}_{1}^{\ast} (f)^{E}_{B_1\ldots B_{k-2}E}\right\|^2_\varphi
=\frac{k-1}{k+2}\sum_{B_{1},\ldots ,B_{k-2}}\left\| \sum_{C,E}\Theta_{C(B_1}f^{CE}_{B_2\ldots B_{k-2}E)}\right\|_{\varphi}^2\\
=&\frac{1}{(k-1)(k+2)}\sum_{B_{1},\ldots ,B_{k-2}}\left\| \sum_{C,E}\Theta_{CE}f^{CE}_{B_1\ldots B_{k-2}}\right\|_{\varphi}^2\\
\leq& \frac{4}{(k-1)(k+2)}\sum_{B_{1},\ldots ,B_{k-2}}\sum_{E}\left\|\sum_{C}\Theta_{CE}f^{CE}_{B_1\ldots B_{k-2}}\right\|^2_{\varphi}\\
\leq &\frac{4}{(k-1)(k+2)}\sum_{A_1,B_{1},\ldots ,B_{k-1}}\left\|\sum_{C}\Theta_{CB_{1}}f^{CA_{1}}_{B_2\ldots B_{k-1}}\right\|^2_{\varphi},
\end{split}
\end{equation}
by \eqref{ad-Dl} and $\mathscr{C} f=0$.  We use the inequality
$|\sum_Ea_E|^{2}\leq 4\sum_E|a_E|^2 $
in the first inequality and add extra nonnegative terms in the second inequality.
Now we have
\begin{equation}\label{d1starf}
\begin{aligned}
(k-1)\left\| \mathscr{D}_{1}^\ast f \right\|^2_{\varphi}&=(k-1)\left\| \mathcal{D}_{1}^\ast f \right\|^2_{\varphi}-(k-1)\left\|\mathscr{P}_1\mathcal{D}_{1}^{\ast} (f)\right\|^2_{\varphi}\\
&=\bigg\{\sum_{\substack {C,D,A_{1},B_{1},\ldots ,B_{k-1}}} \left\langle \Theta_{CB_{1}} f_{B_{2} \ldots B_{k-1}}^{CA_1},\Theta_{DB_{1}} f_{B_{2} \ldots B_{k-1}}^{DA_1} \right\rangle_\varphi-(k-1)\left\|\mathscr{P}_1\mathcal{D}_{1}^{\ast} (f)\right\|^2_{\varphi}\bigg\}\\
&\qquad+\sum_{s=2}^{k-1} \sum_{\substack {C,D,A_{1},B_{1},\ldots ,B_{k-1}}} \left\langle
\Theta_{CB_{s}} f_{B_{1} \ldots \widehat{B_{s}} \ldots B_{k-1}}^{CA_1},\Theta_{DB_{1}} f_{B_{2} \ldots B_{k-1}}^{DA_1} \right\rangle_\varphi
:=\Sigma_1+\Sigma_2,
\end{aligned}
\end{equation}
by  \eqref{depos-1} and expanding symmetrization as  in \eqref{gonge_fanshu}. Apply \eqref{CDf-1} to $\Sigma_1$ in \eqref{d1starf} to get
\begin{equation}\label{1-sigma1}
\Sigma_1\geq\frac{k-2}{k+2}\sum_{D,A_{1},B_{1},\ldots ,B_{k-1}} \left\| \sum_{C}\Theta_{CB_{1}} f_{B_{2}
\ldots B_{k-1}}^{CA_1} \right\|^2_\varphi\geq0,
\end{equation}
 if $k\geq 2$. To control $\Sigma_2$, we use commutator to change order of differential operator again to get
\begin{eqnarray}\label{Sigma_2 l}
\Sigma_2&=& \sum_{s=2}^{k-1} \sum_{\substack {C,D,A_{1}, B_{1},\ldots ,B_{k-1}}} \left\langle
\nabla^{DB_{1}}\Theta_{CB_{s}} f_{B_{1} \ldots \widehat{B_{s}} \ldots B_{k-1}}^{CA_1}, f_{B_{2} \ldots B_{k-1}}^{DA_1} \right\rangle_\varphi\nonumber\\
&=& \sum_{s=2}^{k-1} \sum_{\substack {C,D,A_{1},B_{1},\ldots ,B_{k-1}}} \left\langle
\Theta_{CB_{s}}\nabla^{DB_{1}} f_{B_{1} \ldots \widehat{B_{s}} \ldots B_{k-1}}^{CA_1}, f_{B_{2} \ldots B_{k-1}}^{DA_1} \right\rangle_\varphi\nonumber\\
&\ &+\sum_{s=2}^{k-1} \sum_{\substack {C,D,A_{1},B_{1},\ldots ,B_{k-1}}}\left\langle
\bigg[\nabla^{DB_{1}},\Theta_{CB_{s}}\bigg] f_{B_{1} \ldots \widehat{B_{s}} \ldots B_{k-1}}^{CA_1}, f_{B_{2} \ldots B_{k-1}}^{DA_1} \right\rangle_\varphi
 := \Sigma_3+\Sigma_4,
\end{eqnarray}
by using Lemma \ref{fenbujifen}. As in the case $l=1$, we have
\begin{eqnarray}\label{Sigma_4 l}
\Sigma_4&=&4(k-2)\left\| f \right\|^2_{\varphi}.
\end{eqnarray}

To control $\Sigma_3$, let us isolate the term concerning $\mathscr{D}_2f$. Rewrite $\Sigma_3$ as
\begin{equation}\label{sigma3-2}
\begin{aligned}
\frac{1}{k-2}\Sigma_3=&\frac{1}{k-2}\sum_{s=2}^{k-1} \sum_{\substack {C,D,A_{1},B_{1},\ldots ,B_{k-1}}}\left\langle
\nabla^{DB_{1}} f_{B_{1} \ldots \widehat{B_{s}} \ldots B_{k-1}}^{CA_1}, \nabla^{CB_{s}}f_{B_{2} \ldots B_{k-1}}^{DA_1} \right\rangle_\varphi\\
=&\sum_{\substack{C,D,A_{1},E_1,E_2,\\B_{3}, \ldots ,B_{k-1}}}\left\langle
 \nabla^{DE_1} f_{B_{3} \ldots B_{k-1}E_1}^{CA_1},\nabla^{CE_2}f_{B_{3} \ldots B_{k-1}E_2}^{DA_1} \right\rangle_\varphi\\
 =&\sum_{\substack{C,D,A_{1},E_1,E_2,\\B_{3}, \ldots ,B_{k-1}}}\bigg\{-3\left\langle
 \nabla^{DE_1} f_{B_{3} \ldots B_{k-1}E_1}^{CA_1},\nabla^{E_2[C}f_{B_{3} \ldots B_{k-1}E_2}^{DA_1]} \right\rangle_\varphi\\
 &\quad\quad\quad\quad\quad\quad\,\, +\left\langle \nabla^{DE_1} f_{B_{3} \ldots B_{k-1}E_1}^{CA_1},\nabla^{A_1E_2}f_{B_{3} \ldots B_{k-1}E_2}^{DC}\right\rangle_\varphi\\
 &\quad\quad\quad\quad \quad\quad\,\,+ \left\langle \nabla^{DE_1} f_{B_{3} \ldots B_{k-1}E_1}^{CA_1},\nabla^{DE_2}f_{B_{3} \ldots B_{k-1}E_2}^{CA_1} \right\rangle_\varphi\bigg\}
 :=\Sigma_{31}+\Sigma_{32}+\Sigma_{33},
 \end{aligned}
\end{equation}
by using Lemma \ref{fenbujifen}, Lemma \ref{lem:symzhankai} (2) and relabeling indices.
It is easy to see $\Sigma_{33}$ is a squared sum, which is nonnegative and
 \begin{eqnarray*}
\Sigma_{31} &=&-3\sum_{\substack{C,D,A_1,E_1,E_2,\\B_{3}, \ldots ,B_{k-1}}} \left\langle
 \nabla^{E_1[C}f_{B_{3} \ldots B_{k-1}E_1}^{DA_1]},\nabla^{E_2[C}f_{B_{3} \ldots B_{k-1}E_2}^{DA_1]}\right\rangle_\varphi\\
 &=&-3 \sum_{\substack{C,D,A_1,\\B_{3}, \ldots ,B_{k-1}}} \left\|
\mathscr{D}_2(f)_{B_{3} \ldots B_{k-1}}^{CDA_1}\right\|_\varphi^2=-3 \left\|\mathscr{D}_2f\right\|^2_{\varphi}.
 \end{eqnarray*}
It follows from the expression of $\frac{1}{k-2}\Sigma_3$ in the second identity in \eqref{sigma3-2} that
  \begin{eqnarray*}
 \Sigma_{32}=-\sum_{\substack{D,F,G,E_1,E_2,\\B_{3}, \ldots ,B_{k-1}}} \left\langle
 \nabla^{DE_1} f_{B_{3} \ldots B_{k-1}E_1}^{GF},\nabla^{GE_2}f_{B_{3} \ldots B_{k-1}E_2}^{DF}\right\rangle_\varphi=\frac{-1}{k-2}\Sigma_{3},
 \end{eqnarray*}
 by relabeling indices $A_1$ as $F$ and $C$ as $G$ and using $f^{AB}_{\ldots}$ antisymmetric in $A, B$.
 Hence, we get
\begin{equation}\label{1-sigma3}
\Sigma_3\geq-\frac{3(k-2)}{2}\left\|\mathscr{D}_2f\right\|^2_{\varphi},
\end{equation}
when $k>2$. By \eqref{d1starf}-(\ref{1-sigma3}), we get
\begin{eqnarray*}
\left\| f \right\|_{\varphi}^2\leq \frac{k-1}{4(k-2)}\left\|\mathscr{D}_{1}^\ast f \right\|_{\varphi}^2  +\frac{3}{8}\left\| \mathscr{D}_{2}f \right\|_{\varphi}^2.
\end{eqnarray*}

$(3)$ For $l=3$, since $\mathscr{D}_{3}f=0$, we need to prove $\left\| f\right\|^2\leq C \left\| \mathscr{D}_{2}^{\ast}f\right\|^2$.
Similar to \eqref{depos-1}, we have
\begin{equation}
\left\|\mathcal{D}_{2}^{\ast} f\right\|_\varphi^2=\left\|\mathscr{D}_{2}^{\ast} f\right\|_\varphi^2+\left\|\mathscr{P}_2\mathcal{D}_{2}^{\ast}(f)\right\|_\varphi^2.
\end{equation}
Apply Proposition \ref{prop:pf=some-cf} to $\left\|\mathscr{P}_2\mathcal{D}_{2}^{\ast}(f)\right\|_\varphi^2$ to get
\begin{equation}\label{CDf-2-l=2}
\begin{aligned}(k-2)
 \|\mathscr{P}_2&\mathcal{D}_{2}^{\ast} f \|^2_\varphi=\frac{2(k-2)^2 }{k}\sum_{A_1, B_{1},\ldots ,B_{k-3}}\left\|\sum_{C,E}\Theta_{C(E}f^{CA_1E}_{B_1\ldots B_{k-3})}\right\|^2_\varphi
\\=& \frac{2}{k }\sum_{A_1, B_{1},\ldots ,B_{k-3}}\left\|\sum_{C,E}\Theta_{CE}f^{CA_1E}_{B_1\ldots B_{k-3}}\right\|^2_\varphi
\leq\frac{6}{k }\sum_{A_1,  \ldots ,B_{k-3}}\sum_{E\neq A_1}\left\|\sum_{C}\Theta_{CE}f^{CA_1E}_{B_1\ldots B_{k-3}}\right\|_\varphi^2\\
\leq&\frac{6}{k }\sum_{A_1,A_2, \ldots ,B_{k-2}}\left\|\sum_{C}\Theta_{CB_{1}}f^{CA_1A_2}_{B_2\ldots B_{k-2}}\right\|_\varphi^2,
\end{aligned}
\end{equation}
by using $\mathscr{C}f=0$ again, where we use the inequality
 $|\sum_{j=1}^3a_j|^{2}\leq 3\sum_{j=1}^3|a_j|^2 $ in the first inequality and add some nonnegative terms in the second inequality.

As in the case $l=2$ in \eqref{d1starf}, we have
\begin{equation}\label{d1starf-l=2}
\begin{aligned}
(k-2)\left\| \mathscr{D}_{2}^\ast f \right\|_\varphi^2=&(k-2)\left\| \mathcal{D}_{2}^\ast f \right\|_\varphi^2
   -(k-2)\left\|\mathscr{P}\mathcal{D}_{2}^{\ast}(f)\right\|_\varphi^2\\
 =&\bigg\{\sum_{\substack {C,D,A_{1},A_2,\\B_{1},\ldots ,B_{k-2}}} \left\langle \Theta_{CB_{1}} f_{B_{2} \ldots B_{k-2}}^{CA_1A_2},
\Theta_{DB_{1}} f_{B_{2} \ldots B_{k-2}}^{DA_1A_2} \right\rangle_\varphi-(k-2)\left\|\mathscr{P}\mathcal{D}_{2}^{\ast}(f)\right\|_\varphi^2\bigg\}\\
   &\,\,\,\,+\sum_{s=2}^{k-2} \sum_{\substack {C,D,A_{1},A_2,\\B_{1},\ldots ,B_{k-2}}} \left\langle \Theta_{CB_{s}} f_{B_{1} \ldots \widehat{B_{s}} \ldots B_{k-2}}^{CA_1A_2},
\Theta_{DB_{1}} f_{B_{2} \ldots B_{k-2}}^{DA_1A_2} \right\rangle_\varphi:=\Sigma_1+\Sigma_2.
\end{aligned}
\end{equation}
Apply  \eqref{CDf-2-l=2} to $\Sigma_1$ in \eqref{d1starf-l=2} to get
\begin{equation}\label{sigma2-l=2}
\Sigma_1\geq\bigg(1-\frac{6}{k}\bigg)\sum_{\substack { D,A_{1},A_2,B_{1},\ldots }} \left\|\sum_{C}\Theta_{CB_{1}}f^{CA_1A_2}_{B_2\ldots B_{k-2}}\right\|_\varphi^2\geq0,
\end{equation}
if $k\geq6$. For $\Sigma_2$, we can rewrite it as
\begin{eqnarray}
\Sigma_2&=& \sum_{s=2}^{k-2} \sum_{\substack {C,D,A_{1},A_2, B_{1},\ldots  }}\left\langle
\nabla^{DB_{1}}\Theta_{CB_{s}} f_{B_{1} \ldots \widehat{B_{s}} \ldots B_{k-2}}^{CA_1A_2}, f_{B_{2} \ldots B_{k-2}}^{DA_1A_2} \right\rangle_\varphi\nonumber\\
&=& \sum_{s=2}^{k-2} \sum_{\substack {C,D,A_{1},A_2,B_{1},\ldots  }} \left\langle
\Theta_{CB_{s}}\nabla^{DB_{1}} f_{B_{1} \ldots \widehat{B_{s}} \ldots B_{k-2}}^{CA_1A_2}, f_{B_{2} \ldots B_{k-2}}^{DA_1A_2} \right\rangle_\varphi\nonumber\\
&\ &+\sum_{s=2}^{k-2} \sum_{\substack {C,D,A_{1},A_2,B_{1},\ldots  }} \left\langle
[\nabla^{DB_{1}},\Theta_{CB_{s}}] f_{B_{1} \ldots \widehat{B_{s}} \ldots B_{k-2}}^{CA_1A_2}, f_{B_{2} \ldots B_{k-2}}^{DA_1A_2} \right\rangle_\varphi\nonumber
 := \Sigma_3+\Sigma_4,\label{sigma1-l=2}
\end{eqnarray}
by Lemma \ref{fenbujifen}. Similarly to the case $l=1$, we have
\begin{eqnarray}\label{sigma4-l=2}
\Sigma_4&=&4(k-3)\left\| f \right\|_{\varphi}^{2}.
\end{eqnarray}
To control $\Sigma_3$, we write
\begin{equation*}\label{sigma3=1+2+3+4}
\begin{aligned}
\frac{1}{k-3}\Sigma_3=&\frac{1}{k-3}\sum_{s=2}^{k-2} \sum_{\substack {C,D,A_{1},A_2,B_{1},\ldots ,B_{k-2}}} \left\langle
\nabla^{DB_{1}} f_{B_{1} \ldots \widehat{B_{s}} \ldots B_{k-2}}^{CA_1A_2}, \nabla^{CB_{s}}f_{B_{2} \ldots B_{k-2}}^{DA_1A_2} \right\rangle_\varphi\\
=&\sum_{\substack{C,D,A_{1},A_2,E_1,E_2,\\B_{1}, \ldots ,B_{k-4}}} \left\langle \nabla^{DE_1} f_{B_{1} \ldots B_{k-4}E_1}^{CA_1A_2},
 \nabla^{CE_2}f_{B_{1} \ldots B_{k-4}E_2}^{DA_1A_2} \right\rangle_\varphi\\
 =&\sum_{\substack{C,D,A_{1},A_2,E_1,E_2,\\B_{1}, \ldots ,B_{k-4}}} \bigg\{-4\left\langle
 \nabla^{DE_1} f_{B_{1} \ldots B_{k-4}E_1}^{CA_1A_2},\nabla^{E_2[C}f_{B_{1} \ldots B_{k-4}E_2}^{DA_1A_2]} \right\rangle_\varphi\\
  &\quad\quad\quad\quad\quad\quad\quad+\left\langle \nabla^{DE_1} f_{B_{1} \ldots B_{k-4}E_1}^{CA_1A_2},\nabla^{ A_1 E_2}f_{B_{1} \ldots B_{k-4}E_2}^{DCA_2}\right\rangle_\varphi\\
  &\quad\quad\quad\quad\quad\quad\quad+\left\langle \nabla^{DE_1} f_{B_{1} \ldots B_{k-4}E_1}^{CA_1A_2},\nabla^{ A_2 E_2}f_{B_{1} \ldots B_{k-4}E_2}^{DA_1C}\right\rangle_\varphi\\
  &\quad\quad\quad\quad\quad\quad\quad+\left\langle \nabla^{DE_1} f_{B_{1} \ldots B_{k-4}E_1}^{CA_1A_2},\nabla^{ D E_2}f_{B_{1} \ldots B_{k-4}E_2}^{CA_1A_2}\right\rangle_\varphi\bigg\}
 :=\Sigma_{31}+\Sigma_{32}+\Sigma_{33}+\Sigma_{34},
\end{aligned}
\end{equation*}
by Lemma \ref{fenbujifen} and  relabeling   indices.
It is easy to see that $\Sigma_{34}$ is a nonnegative squared norm, and
 \begin{equation*}\label{sigma31}
 \begin{aligned}
\Sigma_{31} &=&-4\sum_{\substack{C,D,A_1,A_2,E_1,E_2,\\B_{1}, \ldots ,B_{k-4}}} \left\langle \nabla^{E_1[C}f_{B_{1} \ldots B_{k-4}E_1}^{DA_1A_2]},
 \nabla^{E_2[C}f_{B_{1} \ldots B_{k-4}E_2}^{DA_1A_2]}\right\rangle_\varphi=-4\left\|\mathscr{D}_3f\right\|^2_\varphi=0,
 \end{aligned}
 \end{equation*}
by Lemma \ref{lem:sym}, while
  \begin{eqnarray*}\label{sigma32}
 \Sigma_{32}=\Sigma_{33}=\frac{-1}{k-3}\Sigma_3,
 \end{eqnarray*}
by relabeling indices again.  Hence,
\begin{equation}\label{sigma3-l=2}
\Sigma_3\geq0.
\end{equation}
Apply \eqref{sigma2-l=2}-\eqref{sigma3-l=2} to \eqref{d1starf-l=2} to get
\begin{eqnarray*}
\left\| f \right\|_{\varphi}^2\leq \frac{k-2}{4(k-3)}\left\|\mathscr{D}_{2}^\ast f \right\|_{\varphi}^2.
\end{eqnarray*}
The estimate \eqref{k-monogenic l2estimate} is proved.
\qed

\section{Proof  of main theorems}\label{section4}
We use a general machine to deduce the existence  of solution from the $L^2$-estimate (cf. e.g. \cite{ChenShaw}).

\begin{pro}\label{box-dense}
The $\Box_l$ is a densely defined, closed, self-adjoint and non-negative operator with  domain
\begin{equation*}
 Dom(\Box_l)=\{f\in L^2_\varphi(\Bbb R^6,\mathscr{V}_l)|f\in Dom(\mathscr{D}_{l}),f\in Dom(\mathscr{D}_{l-1}^\ast),
\mathscr{D}_{l-1}^\ast f\in Dom(\mathscr{D}_{l-1}),\mathscr{D}_{l} f\in Dom(\mathscr{D}_{l}^\ast)\}.
 \end{equation*}
\end{pro}
This general fact from functional analysis  essentially dues to Gaffney \cite{gaffney} (See also \cite[Proposition 4.2.3]{ChenShaw}  \cite[Proposition 3.1]{Wang2017}). So we omit its proof here.
\vskip 2mm
\emph{Proof of Theorem \ref{Thm-solution}}.
$(1)$ Theorem \ref{Thm:k-monogenic l2e} implies that
$$\frac{1}{C}\left\|h\right\|^2_\varphi\leq\left\|\mathscr{D}_{l-1}^\ast h\right\|^2_\varphi+\left\|\mathscr{D}_{l}h\right\|^2_\varphi=\langle\Box_l h,h\rangle_\varphi\leq
\left\|\Box_l h\right\|_\varphi\left\|h\right\|_\varphi,$$
for $h\in Dom(\Box_l)$.
Thus $\Box_l$ is bounded from below and injective. Since $\Box_l$ is self-adjoint and closed, $Range\,\, \Box_l$ is a dense subset of $L^2_\varphi(\Bbb R^6, \mathscr{V}_{l})$ by Proposition \ref{box-dense}. For fixed $f\in L^2_\varphi(\Bbb R^6, \mathscr{V}_{l})$, we define the complex anti-linear functional
$$\lambda_f:\Box_l h\longrightarrow\langle f,h\rangle_\varphi,$$
which is well defined on the dense subset $Range \,\,\Box_l$ of $L^2_\varphi(\Bbb R^6, \mathscr{V}_{l})$, since
$$|\lambda_f(\Box_l h)|=|\langle f,h\rangle_\varphi|\leq\left\|f\right\|_\varphi\left\|h\right\|_\varphi\leq C\left\|f\right\|_\varphi\left\|\Box_l h\right\|_\varphi,$$
for  $h\in Dom \Box_l$. We  see that $\lambda_f$ is bounded on a dense subset and can be uniquely extended to the whole space $L^2_\varphi(\Bbb R^6, \mathscr{V}_{l})$. By the Riesz representation theorem, there exists a unique $F\in L^2_\varphi(\Bbb R^6, \mathscr{V}_{l})$ such that $\lambda_f(G)=\langle F,G\rangle_\varphi$ for any $G\in L^2_\varphi(\Bbb R^6, \mathscr{V}_{l})$ and $\left\|F\right\|_\varphi=|\lambda_f|\leq C\left\|f\right\|_\varphi$. So we have $\langle F,\Box_l h\rangle_\varphi=\langle f,h\rangle_\varphi$ for any $h\in Dom(\Box_l)$. This implies $F\in Dom(\Box_l^\ast)$ and $\Box_l^\ast F=f$. Since $\Box_l$ is self-adjoint, $F\in Dom(\Box_l)$ and $\Box_lF=f$.
We write $F= N_lf$. Then $\left\|N_lf\right\|_\varphi\leq C\left\|f\right\|_\varphi$.

$(2)$ Since $N_{l+1}f\in Dom(\Box_{l+1})$, we have $\mathscr{D}_{l}^\ast N_{l+1}f\in Dom(\mathscr{D}_{l}),\mathscr{D}_{l+1}N_{l+1}f\in Dom(\mathscr{D}_{l+1}^\ast)$ and
\begin{equation}\label{decom}
\mathscr{D}_{l}\mathscr{D}_{l}^\ast N_{l+1}f=f-\mathscr{D}_{l+1}^\ast\mathscr{D}_{l+1} N_{l+1}f,
\end{equation}
by $\Box_{l+1} N_{l+1} f=f$.
Because $\mathscr{D}_{l+1}f=0$ and $\mathscr{D}_{l+1}\mathscr{D}_{l}H=0$ for any $H\in Dom (\mathscr{D}_{l})$, the above identity implies $\mathscr{D}_{l+1}^\ast\mathscr{D}_{l+1} N_{l+1}f \in Dom(\mathscr{D}_{l+1})$ and
$$ \mathscr{D}_{l+1}\mathscr{D}_{l+1}^\ast\mathscr{D}_{l+1} N_{l+1}f=0,$$
 by $\mathscr{D}_{l+1}$ acting on both sides of \eqref{decom}. Then
$$0=\langle\mathscr{D}_{l+1}\mathscr{D}_{l+1}^\ast\mathscr{D}_{l+1} N_{l+1}f, \mathscr{D}_{l+1} N_{l+1}f\rangle_\varphi=\left\|\mathscr{D}_{l+1}^\ast\mathscr{D}_{l+1} N_{l+1}f\right\|^2_\varphi,$$
i.e., $\mathscr{D}_{l+1}^\ast\mathscr{D}_{l+1} N_{l+1}f=0$. Hence, by \eqref{decom}, we have
$$\mathscr{D}_{l}\mathscr{D}_{l}^\ast N_{l+1}f=f.$$
Moreover, we have $\mathscr{D}_{l}^\ast N_{l+1}f \perp \ker \mathscr{D}_{l}$ since
$\langle H,\mathscr{D}_{l}^\ast N_{l+1}f\rangle_\varphi=\langle\mathscr{D}_{l}H,N_{l+1}f\rangle_\varphi=0$ for any $H\in \ker\mathscr{D}_{l}$. The estimate \eqref{solution-estimate} follows from
$$\left\|\mathscr{D}_{l}^\ast N_{l+1}f\right\|_\varphi^2+\left\|\mathscr{D}_{l+1}N_{l+1}f\right\|_\varphi^2=\langle\Box_{l+1}N_{l+1}f,N_{l+1}f\rangle_\varphi\leq C\left\|f\right\|_\varphi^2.$$
The theorem is proved. \qed
 \vskip 2mm

\emph{Proof of Theorem \ref{poly-exact}}.
Note that $\left\|f\right\|_{\varphi}^2 < +\infty$ for $ f \in P(\Bbb R^{6},\mathscr{V}_{l+1})$, where  $\varphi=\left|x\right|^2$. So there exists $u \in L^2(\Bbb R^{6},\mathscr{V}_{l})$, such that
$\mathscr{D}_{l} u=f$ and $\mathscr{D}_{l-1}^\ast u=0$ by Theorem \ref{Thm-solution}. Consequently,
\begin{equation}\label{eq:box u}
\Box_{l}u=\Theta_{l-1} \mathscr{D}_l u+\mathscr{D}_{l-1}\Theta_{l-1}  u=\Theta_{l-1} f, \,\,\,\,\,\,\,\,l=1,2,3.
\end{equation}
in the sense of distributions, where $\mathscr{D}_{l+1}f=0$ and $\Theta_{l-1} f$ is a polynomial    by the expression of $ \mathcal{{D}}_{l}^*$ in \eqref{ad-Dl} and $\mathscr{P}$ in \eqref{Pf}.

On the other hand, $\Box_{l}$ is an elliptic differential operator of second order. This is because
 $$\langle \sigma(\Box_{l}) \xi,\xi\rangle=\langle\sigma_l \xi,\sigma_l \xi\rangle+\langle\sigma_{l-1}^\ast \xi,\sigma_{l-1}^\ast \xi\rangle,$$
 for  $\xi\in \mathscr{V}_{l}$,
where the inner product is the Euclidean inner product of $\mathscr{V}_{l}$ and $\sigma(\Box_{l})$ and $\sigma_l$ are symbols of   operators
 $\Box_{l}$ and $\mathscr{D}_l$ (cf. \eqref{def:symbol}), respectively. We see that
$$\ker\sigma(\Box_{l})=\ker\sigma_l\cap \ker\sigma_{l-1}^\ast={\rm Im}\sigma_{l-1}\cap \ker\sigma_{l-1}^\ast=\{0\},$$
 by Proposition \ref{pro:elliptic complex}.
 Thus  we know the solution $u$ of \eqref{eq:box u}   is real analytic by applying Theorem 6.6.1 in \cite{morrey} to elliptic differential operator $\Box_{l}$ of second order with real analytic coefficients. We write the Taylor expression of $u$
as $
u=\sum_{m=0}^{\infty}u_m,
$
 where $u_m$ is a polynomial of homogeneous degree $m$. Suppose $f$ is a polynomial of degree $L$. Since $\mathscr{D}_l$ is a first order differential operator with constant coefficients, then $\mathscr{D}_l u_m$ is a polynomial of degree $m-1$ or vanishes. Hence, $\mathscr{D}_l u=f$ implies that
  $$\mathscr{D}_l \bigg(\sum_{m=0}^{L+1}u_m\bigg)=f.$$
So we get a polynomial solution to $\mathscr{D}_l u=f$ if $\mathscr{D}_{l+1} f=0$. The result follows. \qed

\section{The ellipticity of $k$-monogenic-complex}\label{section5}
Recall that the symbol of the matrix differential operator $\mathcal{D}=\sum_{|\alpha|\leq m}A_{\alpha_1\ldots\alpha_N}(x)\partial_{x_1}^{\alpha_1}\ldots\partial_{x_N}^{\alpha_N}: C^{\infty}(\Omega, W)\longrightarrow C^{\infty}(\Omega, W^{\prime})$ at $(x,v)$
  is defined to be
\begin{equation}\label{def:symbol}
\sigma(\mathcal{D})_{(x,v)}:=\sum_{|\alpha|=m}A_{\alpha_1\ldots\alpha_N}(x)\bigg(\frac{v_1}{i}\bigg)^{\alpha_1}
\ldots\bigg(\frac{v_N}{i}\bigg)^{\alpha_N}:W\longrightarrow W^{\prime},
\end{equation}
where $\Omega$ is a domain in  $\mathbb{R}^N$ and $A_{\alpha_1\ldots\alpha_N}$ is a linear transformation from vector space $W$ to $W^{\prime}$, $v\in \Bbb R^{N}$.
A differential complex
$$C^\infty(\Omega,W_0)\overset{\mathcal{D}_0}\longrightarrow\ldots\overset{\mathcal{D}_{n-1}}\longrightarrow C^\infty(\Omega,W_n)$$
 is called \emph{elliptic} if its symbol sequence
$$W_0\overset{\sigma(\mathcal{D}_0)_{(x,v)}}\longrightarrow\ldots\overset{\sigma(\mathcal{D}_{n-1})_{(x,v)}}\longrightarrow W_n$$
is exact for any $x\in\Omega$, $v\in\mathbb{R}^N\backslash\{0\}$, that is
 $\ker\sigma(\mathcal{D}_l)_{(x,v)}={\rm Im}\, \sigma(\mathcal{D}_{l-1})_{(x,v)}.$

\emph{Proof of Proposition \ref{pro:elliptic complex}}.\,\,Let us prove the symbol sequence
\begin{equation}\label{sigma-exact}
0\longrightarrow
\mathscr{V}_0\overset{\sigma_0}\longrightarrow\ldots\overset{\sigma_{2}}\longrightarrow
\mathscr{V}_3\longrightarrow 0,
\end{equation}
is exact  for fixed $x\in\Bbb{R}^6$ and $v\in\Bbb{R}^6\backslash\{0\}$,  where $\sigma_l:=\sigma(\mathscr{D}_l)_{(x,v)}$. Note that
\begin{equation}\label{sigmaD}
(\sigma_lf)^{A_1 \ldots A_{l+1}}_{B_2 \ldots B_{k-l}}=\sum_{B_{1}=1}^{4}M^{B_{1}[A_1}f^{A_2 \ldots
A_{l+1}]}_{B_1 \ldots B_{k-l}},
\end{equation}
with $$
 M^{AB}:=\frac{1}{i}\left(
 \begin{matrix}
   0 & iv_0+v_5 & v_3+iv_4 & v_1+iv_2 \\
   -iv_0-v_5 & 0 & v_1-iv_2 & -v_3+iv_4 \\
   -v_3-iv_4 & -v_1+iv_2 & 0 & -iv_0+v_5 \\
   -v_1-iv_2 & v_3-iv_4 &iv_0-v_5 & 0
  \end{matrix}
  \right)
$$
a antisymmetric matrix. Since $\sigma_{l+1}\circ\sigma_l=0$ follows from $\mathscr{D}_{l+1}\circ \mathscr{D}_l=0$,  we only need to prove
$\sigma_0$ is injective, $\ker\sigma_l\subseteq Im \sigma_{l-1}$, $l=1,2$, and $\sigma_2$ is surjective. \par

$(1)$ For any $\xi\in \ker\sigma_0$, we have
 \begin{equation*}
 \sigma_0(\xi)^{A_1}_{B_2 \ldots B_{k}}=\sum_{B_{1} } M^{B_1 A_1} \xi_{B_1 B_2\ldots B_k}=0,
 \end{equation*}
for any fixed $A_1,B_2, \ldots, B_{k}$. It is known that the determinant of $M$ is nonvanishing for any $v\neq 0$ since $M\overline{M}^T=|v|^2 I_{4\times 4}$ which can be deduced from  \cite[(2.5) and Proposition 2.1]{wangkang3}. This essentially comes from the fact that $\mathscr{D}_0$ is the Dirac operator. So we have $\xi_{B_1 \ldots B_{k}}=0$ for any $B_1,\ldots, B_k$.
 Hence, $\sigma_0$ is injective. \par

 $(2)$ For any $\xi\in \ker\sigma_1$, let $\Xi\in \mathscr{V}_0$ be given by
$$\Xi_{B_1\ldots B_{k}}:=\sum_{E } M^{-1}_{E(B_{1}}\xi^{E}_{B_2\ldots B_{k})},$$
where $M^{-1}$ is the inverse of $M$. Then
\begin{equation}\label{sigma0F}
\begin{aligned}
\sigma_{0}(\Xi)^{A_1}_{B_2\ldots B_{k}}=&\sum_{B_{1}}M^{B_{1}A_1}\Xi_{B_1B_2\ldots B_{k}}\\
=&\frac{1}{k}\sum_{E,B_{1}}\bigg[ M^{-1}_{EB_{1}}M^{B_{1}A_1}\xi^{E}_{B_2\ldots B_{k}}
+\sum_{s=2}^{k}M^{-1}_{EB_{s}}M^{B_{1}A_1}\xi^{E}_{B_1\ldots \widehat{B_s}\ldots B_{k}}\bigg],
\end{aligned}
\end{equation}
by using Lemma \ref{lem:symzhankai} (1). Since $\xi\in \ker\sigma_1$,   we have
\begin{equation}\label{l=1-antisym}
\begin{aligned}
0=2\sigma_{1}(\xi)^{A_1E}_{B_2\ldots \widehat{B_s}\ldots B_{k}}
=\sum_{B_{1}}M^{B_{1}A_1}\xi^{E}_{B_1\ldots \widehat{B_s}\ldots B_{k}}-\sum_{B_{1}}M^{B_{1}E}\xi^{A_1}_{B_1\ldots \widehat{B_s}\ldots B_{k}},
\end{aligned}
\end{equation}
by Lemma \ref{lem:symzhankai} (2). Apply \eqref{l=1-antisym} to \eqref{sigma0F} to get
\begin{eqnarray*}
\sigma_{0}(\Xi)^{A_1}_{B_2\ldots B_{k}}&=&\frac{1}{k}\bigg[\sum_{E}\delta_{E}^{A_1}\xi^{E}_{B_2\ldots
B_{k}}+\sum_{s=2}^{k}\sum_{B_{1}}\delta_{B_s}^{B_{1}}\xi^{A_1}_{B_1\ldots \widehat{B_s}\ldots B_{k}}\bigg] =\xi^{A_1}_{B_2\ldots B_{k}},
\end{eqnarray*}
since $M^{-1}$ is the inverse of $M$. Thus $\sigma_{0}\Xi=\xi$ and so $\ker \sigma_1\subseteq {\rm Im}\sigma_0$.\par

$(3)$ For any $\xi\in \ker\sigma_2$, set
$$\Xi^{A}_{B_1\ldots B_{k-1}}:=\sum_{E } M^{-1}_{E(B_{1}}\xi^{A E}_{B_2\ldots B_{k-1})}.$$
We claim $\Xi\in \mathscr{V}_{1}$. Then
\begin{equation}\label{sigma1-F}
\begin{aligned}
\sigma_{1}(\Xi)^{A_1A_2}_{B_2\ldots B_{k-1}}=&\sum_{B_{1}}M^{B_{1}[A_1}\Xi^{A_2]}_{B_1B_2\ldots B_{k-1}}\\
=&\frac{1}{k-1}\sum_{E,B_{1}}\bigg[M^{-1}_{EB_{1}}M^{B_{1}[A_1}\xi^{A_2]E}_{B_2\ldots B_{k-1}}+\sum_{s=2}^{k-1}M^{-1}_{EB_{s}}M^{B_{1}[A_1}\xi^{A_2]E}_{B_1\ldots \widehat{B_s}\ldots B_{k-1}}\bigg],
 \end{aligned}
\end{equation}
by Lemma \ref{lem:symzhankai} (1). Since $\xi\in \ker\sigma_2$, then for fixed $s\in \{2,\ldots,k-1\}$, we have
\begin{equation}\label{antisym-2}
\begin{aligned}
0=&3\sigma_2 (\xi)^{A_1A_2E}_{B_2\ldots \widehat{B_s}\ldots B_{k-1}}=3\sum_{B_{1}}M^{B_{1}[A_1}\xi^{A_2 E]}_{B_1\ldots \widehat{B_s}\ldots B_{k-1}}\\
=&2\sum_{B_{1}}M^{B_{1}[A_1}\xi^{A_2]E}_{B_1\ldots \widehat{B_s}\ldots B_{k-1}}+\sum_{B_{1}}M^{B_{1}E}\xi^{A_1A_2}_{B_1\ldots \widehat{B_s}\ldots B_{k-1}},
\end{aligned}
\end{equation}
by \eqref{sigmaD} and Lemma \ref{lem:symzhankai} (2). Apply \eqref{antisym-2} to \eqref{sigma1-F} to get
\begin{eqnarray*}
\sigma_{1}(\Xi)^{A_1 A_2}_{B_2\ldots B_{k-1}}&=&\frac{1}{k-1}\bigg(\sum_{E}\delta_{E}^{[A_1}\xi^{A_2]E}_{B_2\ldots
B_{k-1}}-\frac{1}{2}\sum_{s=2}^{k-1}\sum_{B_{1}}\delta_{B_s}^{B_{1}}\xi^{A_1A_2}_{B_1\ldots \widehat{B_s}\ldots B_{k-1}}\bigg)
=\frac{-k}{2(k-1)}\xi^{A_1A_2}_{B_2\ldots B_{k-1}},
\end{eqnarray*}
 by $M^{-1}$   inverse to $M$ again. Thus
  $\sigma_{1}\bigg(\frac{2(k-1)}{-k}\Xi\bigg)=\xi.$

  It remains to show the claim $\mathscr C (\Xi)=0$. Note that for any fixed $B_1,\ldots, B_{k-2}$,
\begin{equation}\label{g2inV1}
\begin{aligned}
(k-1)\mathscr C (\Xi)_{B_1\ldots B_{k-2}}&=(k-1)\sum_{A_1,A_2}M^{-1}_{A_2(B_1}\xi^{A_1A_2}_{B_2\ldots B_{k-2}A_1)}\\
=&\sum_{A_1,A_2}\bigg[M^{-1}_{A_2A_1}\xi^{A_1A_2}_{B_1\ldots B_{k-2}}+\sum_{s=1}^{k-2}M^{-1}_{A_{2}B_s}\xi^{A_1A_2}_{B_1\ldots A_{1} \ldots}\bigg]
=\sum_{A_1,A_2}M^{-1}_{A_2A_1}\xi^{A_1A_2}_{B_1\ldots B_{k-2}},
\end{aligned}
\end{equation}
by $\mathscr C \xi=0$. Since $\det M\neq 0$, $C (\Xi) =0$ follows from
\begin{equation*}\label{f-signa2}
\begin{aligned}(k-1)\sum_{B_{1}}M^{EB_{1}}\mathscr C (\Xi)_{B_1\ldots B_{k-2}}=&\sum_{A_1,A_2,B_{1}}M^{B_{1}E}M^{-1}_{A_2A_1}\xi^{A_1A_2}_{B_1\ldots B_{k-2}}\\
= &-\sum_{A_1,A_2}\sum_{B_{1}}\left(M^{B_{1}A_1}\xi^{A_2 E}_{B_1\ldots B_{k-2}}-M^{B_{1}A_2}\xi^{A_1E}_{B_1\ldots B_{k-2}}\right)M^{-1}_{A_2A_1}\\
=&\sum_{A_2,B_{1}}\delta^{B_{1}}_{A_2}\xi^{A_2 E}_{B_1\ldots B_{k-2}}+\sum_{A_1,B_{1}}\delta_{A_1}^{B_{1}}\xi^{A_1E}_{B_1\ldots B_{k-2}}
=2\sum_{B_{1}}\xi^{B_{1}E}_{B_1\ldots B_{k-2}}=0,
\end{aligned}
\end{equation*}for all indices $E,B_2,\ldots,B_{k-2}$,
by using \eqref{antisym-2},  $M$ antisymmetric and $\mathscr C \xi=0$. So $\Xi\in \mathscr{V}_1$.  $\ker \sigma_{2}\subseteq {\rm Im} \sigma_{1}$ is proved. \par

$(4)$  For any $\xi\in \ker\sigma_3=\mathscr{V}_{3}$, we do not know whether $\sum_{E}M^{-1}_{E(B_{1}}\xi^{A_1A_2 E}_{B_2\ldots B_{k-2})}$  belongs to $\mathscr{V}_2$ or not. But note that the diagram
\begin{equation*}
\begin{CD}\odot^{k-1}\Bbb{C}^4\otimes
\wedge^3\Bbb{C}^4  @>4\widetilde{\sigma}>>    \odot^{k-2}\Bbb{C}^4 \otimes\wedge^4\Bbb{C}^4 @>>>  0\\
@VV\mathscr{C}V                                                               @VV\mathscr{C}V\\
\mathscr{V}_2                                        @>-3\sigma_2>>   \mathscr{V}_3 @>>>  0
\end{CD}
\end{equation*}
is commutative, i.e.,  $-3\sigma_2\mathscr{C} = 4\mathscr{C}\widetilde{\sigma} $, where $\widetilde{\sigma}:\,\,\odot^{k-1}\Bbb{C}^4\otimes\wedge^3\Bbb{C}^4 \longrightarrow \odot^{k-2}\Bbb{C}^4 \otimes\wedge^4\Bbb{C}^4$ is given by
\begin{equation}\label{sigmaD-2}
 (\widetilde{\sigma} \widetilde{\Xi})^{A_1 \ldots A_{4}}_{B_2 \ldots B_{k-1}}=\sum_{B_{1}}M^{B_{1}[A_1}\widetilde{\Xi}^{A_2 \ldots A_{4}]}_{B_1 \ldots B_{k-1}}.
\end{equation} This is because
\begin{equation*}
\begin{aligned}
-3(\sigma_2\mathscr{C}\widetilde{\Xi})^{A_1A_2A_3}_{B_1\ldots B_{k-3}}=&-3\sum_{E,F}M^{E [A_1}\widetilde{\Xi}^{\mid F\mid A_2A_3]}_{B_1\ldots B_{k-3}EF}\\
=&-\sum_{E,F}\bigg(M^{E A_1}\widetilde{\Xi}^{ FA_2A_3}_{B_1\ldots B_{k-3}EF}-M^{E A_2}\widetilde{\Xi}^{F A_1 A_3}_{B_1\ldots B_{k-3}EF}-M^{EA_3}\widetilde{\Xi}^{F A_2 A_1}_{B_1\ldots B_{k-3}EF}\bigg)\\
=&=4\sum_{E,F}M^{E[F}\widetilde{\Xi}^{A_1A_2A_3]}_{B_1\ldots B_{k-3}FE}=4(\mathscr{C}\widetilde{\sigma}\widetilde{\Xi})^{A_1A_2A_3}_{B_1\ldots B_{k-3}}
\end{aligned}
\end{equation*}
by $\sum_{E,F}M^{EF}\widetilde{\Xi}^{A_1A_2A_3}_{B_1\ldots B_{k-3}FE}=0$ since $\xi$ is symmetric in $E,F$ while $M$ is antisymmetric in $E,F$.

Now we construct an inverse image of $\sigma_2$ by an inverse image of $\widetilde{\sigma} $. Suppose that $A_1,\ldots,A_4$ are different. There must be at least one of $A_1, A_2, A_3, A_4$ equal  to one of $B_1,\ldots, B_{k-2}$. Without loss of generality, we assume $A_1=B_{k-2}$.
For $\xi\in \mathscr{V}_{3}$, we construct a lifting $\widetilde{\xi}\in \odot^{k-2}\Bbb{C}^4 \otimes  \wedge^4\Bbb{C}^4$ as follows
\begin{equation}\label{pull up}
\widetilde{\xi}^{A_1A_2A_3A_4}_{B_1\ldots B_{k-2}}=\xi^{A_2A_3A_4}_{B_1\ldots B_{k-3}},
\end{equation}
when $A_1=B_{k-2}$. $\widetilde{\xi}$ is well defined because if there also exists $A_2=B_{k-3}$, we must have $\widetilde{\xi}^{A_1A_2A_3A_4}_{B_1\ldots B_{k-2}}=-\xi^{A_1A_3A_4}_{B_1\ldots B_{k-4}B_{k-2}}$ by $\xi^{A_2A_3A_4}_{B_1\ldots B_{k-3}}=-\xi^{A_1A_3A_4}_{B_1\ldots B_{k-4}B_{k-2}}$. The latter identity follows from
$$0=\sum_{E}\xi^{EA_3A_4}_{B_1\ldots B_{k-4}E}=\sum_{E=A_1,A_2}\xi^{EA_3A_4}_{B_1\ldots B_{k-4}E},$$
by $\mathscr{C} \xi=0$ for $\xi\in\mathscr{V}_{3}=\ker\sigma_3$. We have
\begin{equation}\label{F-f}
\mathscr C( \widetilde{\xi})^{A_1A_2A_3}_{B_1\ldots B_{k-3}}=\sum_{C}\widetilde{\xi}^{CA_1A_2A_3}_{B_1\ldots B_{k-3} C}=\xi^{A_1A_2A_3}_{B_1\ldots B_{k-3}}
\end{equation}
for any fixed $A_1,A_2,A_3,B_1,\ldots, B_{k-3}$. Now define $\widetilde{\Xi}\in\odot^{k-1}\Bbb{C}^4\otimes \wedge^3\Bbb{C}^4$ by
\begin{equation*}
\widetilde{\Xi}^{E_2A_1A_2}_{B_1\ldots B_{k-1}}:=\sum_{E_1}M^{-1}_{E_1(B_1}\widetilde{\xi}^{E_2 A_1A_2 E_1}_{B_2\ldots B_{k-1})},
\end{equation*}
and
$\Xi:=\mathscr C \widetilde{\Xi}$. Then
\begin{equation*}
\Xi^{A_1A_2}_{B_1\ldots B_{k-2}}=\sum_{E_1,E_2}M^{-1}_{E_1(B_1}\widetilde{\xi}^{E_2 A_1A_2 E_1}_{B_2\ldots B_{k-2}E_2)}.
\end{equation*}
and $\Xi\in \mathscr{V}_2$, since $\mathscr C \circ\mathscr C \widetilde{\Xi}=0$. Now we show $\sigma_{2}\Xi=C\xi$ for some constant $C\neq 0$.
\begin{equation}\label{S2g=f}
\begin{aligned}
(k-1)(\sigma_{2}\Xi)^{A_1A_2A_3}_{B_2\ldots B_{k-2}}=&(k-1)\sum_{B_{1}}M^{B_{1}[A_1}\Xi^{A_2A_3]}_{B_1\ldots B_{k-2}}\\
=&\sum_{B_{1},E_1,E_2}\bigg[M^{B_{1}[A_1}M^{-1}_{E_1E_2}\widetilde{\xi}^{A_2A_3]E_2 E_1}_{B_1\ldots B_{k-2}}
+\sum_{s=1}^{k-2}M^{B_{1}[A_1}M^{-1}_{E_1 B_s}\widetilde{\xi}^{A_2A_3]E_2 E_1}_{B_1\ldots \widehat{B}_s \ldots B_{k-2}E_2}\bigg]\\
=&\sum_{B_{1},E_1,E_2}M^{-1}_{E_1E_2}M^{B_{1}[A_1}\widetilde{\xi}^{A_2A_3]E_2 E_1}_{B_1\ldots B_{k-2}}
-\sum_{B_{1},E_1}M_{E_1B_{1}}^{-1}M^{B_{1}[A_1}\xi^{A_2 A_3]E_1}_{B_2\ldots B_{k-2}}\\
&-\sum_{s=2}^{k-2}\sum_{B_{1},E_1}M^{-1}_{E_1 B_s}M^{B_{1}[A_1}\xi^{A_2A_3]E_1}_{B_1\ldots \widehat{B}_s\ldots B_{k-2}}:=\Sigma_1+\Sigma_2+\Sigma_3,
\end{aligned}
\end{equation}
by expanding symmetrization and using \eqref{F-f}.
It is easy to see that
\begin{equation}\label{sur-sigma3}
\Sigma_2=-\sum_{E_1}\delta_{E_1}^{[A_1}\xi^{A_2 A_3]E_1}_{B_2\ldots B_{k-2}}=-\xi^{A_1A_2 A_3}_{B_2\ldots B_{k-2}}.
\end{equation}

On the other hand, it follows from  $\xi \in \mathscr{V}_{3}=ker\sigma_{3}$, i.e. $4(\sigma_{3}\xi)^{A_1E_1A_2 A_3}_{B_2\ldots\widehat{B}_s \ldots B_{k-2}}=0$, that
\begin{equation*}\label{simga3-f}
0= 4\sum_{B_{1}}M^{B_{1}[A_1}\xi^{E_1A_2A_3]}_{B_1\ldots \widehat{B}_s\ldots B_{k-2}}
= 3\sum_{B_{1}}M^{B_{1}[A_1}\xi^{A_2A_3]E_1}_{B_1\ldots \widehat{B}_s\ldots B_{k-2}}-\sum_{B_{1}}M^{B_{1}E_1}\xi^{A_1A_2A_3}_{B_1\ldots \widehat{B}_s\ldots B_{k-2}}.
\end{equation*}
Apply this identity to $\Sigma_3$ in \eqref{S2g=f} to get
\begin{equation}\label{sur-sigma2}
\begin{aligned}
\Sigma_3=&-\frac{1}{3}\sum_{s=2}^{k-2}\sum_{B_{1},E_1}M^{-1}_{E_1 B_s}M^{B_{1}E_1}\xi^{A_1A_2A_3}_{B_1\ldots \widehat{B}_s\ldots B_{k-2}}
=-\frac{1}{3}\sum_{s=2}^{k-2}\sum_{B_{1}}\delta_{B_s}^{B_{1}}\xi^{A_1A_2A_3}_{B_1\ldots \widehat{B}_s\ldots  B_{k-2}}\\
=&\frac{3-k}{3}\xi^{A_1A_2A_3}_{B_2\ldots B_{k-2}},
\end{aligned}
\end{equation} by $M^{-1}$   inverse to $M$ again.
Note that
 $M^{B_{1}[A_1}\widetilde{\xi}^{A_2A_3 E_1E_2]}_{B_1\ldots B_{k-2}}=0 $
by $\wedge^5\Bbb C^4=\{0\}$, which implies
\begin{equation}\label{M-F}
\begin{aligned}
&M^{B_{1}E_1}\widetilde{\xi}^{A_2A_3 A_1E_2}_{B_1\ldots B_{k-2}}+M^{B_{1}E_2}\widetilde{\xi}^{A_2A_3 E_1A_1}_{B_1\ldots B_{k-2}} =3M^{B_{1}[A_1}\widetilde{\xi}^{A_2A_3]E_1E_2}_{B_1\ldots B_{k-2}}.
\end{aligned}
\end{equation}
Then apply \eqref{M-F} to $\Sigma_1$ in \eqref{S2g=f} to get
\begin{equation}\label{sur-sigma1}
\begin{aligned}
\Sigma_1&=\frac{1}{3}\sum_{B_{1},E_1,E_2}M^{-1}_{E_1E_2}\left(M^{B_{1}E_1}\widetilde{\xi}^{A_2A_3 A_1E_2}_{B_1\ldots
B_{k-2}}+M^{B_{1}E_2}\widetilde{\xi}^{A_2A_3 E_1A_1}_{B_1\ldots B_{k-2}}\right)\\
&=\frac{1}{3}\sum_{B_{1},E_2}\delta^{B_{1}}_{E_2}\widetilde{\xi}^{A_2A_3 A_1E_2}_{B_1\ldots B_{k-2}}
-\frac{1}{3}\sum_{B_{1},E_1}\delta^{B_{1}}_{E_1}\widetilde{\xi}^{A_2A_3 E_1A_1}_{B_1\ldots B_{k-2}} =-\frac{2}{3}\xi^{A_1A_2A_3}_{B_2\ldots B_{k-2}},
\end{aligned}
\end{equation}
by \eqref{F-f}.
Now apply \eqref{sur-sigma3}, \eqref{sur-sigma2} and \eqref{sur-sigma1} to \eqref{S2g=f} to get
\begin{equation*}
\begin{aligned}
(k-1)(\sigma_{2}\Xi)^{A_1A_2A_3}_{B_2\ldots B_{k-2}}=-\frac{k+2}{3}\xi^{A_1A_2A_3}_{B_2\ldots B_{k-2}}.
\end{aligned}
\end{equation*}
Hence, $\sigma_2$ is surjective.   Proposition \ref{pro:elliptic complex} is proved.\qed

\vskip 3mm

{\bf Data availability statement:}
This manuscript has no associated data.

\end{document}